\documentclass[a4paper,11pt]{article}

\usepackage{bbm}
\usepackage{marginnote}
\usepackage{fancyhdr}
\def\conf{0}
\def\soda{0}
\def\withcolors{0}
\def\lipics{0}

\usepackage{etex}

\usepackage[margin=2.54cm]{geometry}
\usepackage[dvipsnames,usenames]{color}
\usepackage{amsfonts,amsmath,amssymb,amsthm}

\usepackage[margin=10pt]{subfig}

\usepackage{paralist}
\usepackage{algorithmic,algorithm}
\usepackage{bm}
\usepackage{xspace}
\usepackage{xspace,prettyref}
\usepackage{color}
\usepackage{dsfont}
\usepackage{bbm}
\usepackage{tkz-euclide}

\usepackage[pagebackref,letterpaper=true,colorlinks=true,pdfpagemode=none,urlcolor=blue,linkcolor=blue,citecolor=BrickRed,pdfstartview=FitH]{hyperref}

\newcommand{\ignore}[1]{}

\newtheorem{theorem}{Theorem}

\newtheorem{lemma}[theorem]{Lemma}
\newtheorem{claim}[theorem]{Claim}
\newtheorem{definition}[theorem]{Definition}

\theoremstyle{definition}

\def\FullBox{\hbox{\vrule width 8pt height 8pt depth 0pt}}
\newcommand{\QED}{\;\;\;\FullBox}
\renewenvironment{proof}{\noindent{\bf Proof:~~}}{\hfill\QED \par\bigskip}
\newcommand{\BPF}{\begin{Proof}} \newcommand {\EPF}{\end{Proof}}

\newenvironment{proofof}[1]{\smallskip\noindent{\bf Proof of #1:}}%
        {\hspace*{\fill}\QED \par\bigskip}

\newcommand{\EX}{{\rm Exp}}

\newcommand{\var}{{\rm Var}}

\renewcommand{\th}{^\textrm{th}}

\def\eqdef{~\triangleq~}
 \def\eps{\varepsilon}
\def\bar{\overline}

\newcommand{\cA}{{\cal A}}
\newcommand{\cB}{{\cal B}}

\newcommand{\cG}{{\cal G}}

\newcommand{\Sec}[1]{\hyperref[sec:#1]{Section~\ref*{sec:#1}}} %section
\newcommand{\Eq}[1]{\hyperref[eq:#1]{Equation~(\ref*{eq:#1})}} %equation
\newcommand{\Fig}[1]{\hyperref[fig:#1]{Fig.\,\ref*{fig:#1}}} %figure
\newcommand{\Tab}[1]{\hyperref[tab:#1]{Tab.\,\ref*{tab:#1}}} %table
\newcommand{\Thm}[1]{\hyperref[thm:#1]{Theorem\,\ref*{thm:#1}}} %theorem
\newcommand{\Lem}[1]{\hyperref[lem:#1]{Lemma\,\ref*{lem:#1}}} %lemma
\newcommand{\Prop}[1]{\hyperref[prop:#1]{Prop.~\ref*{prop:#1}}} %property
\newcommand{\Cor}[1]{\hyperref[cor:#1]{Corollary~\ref*{cor:#1}}} %corollary
\newcommand{\Def}[1]{\hyperref[def:#1]{Definition~\ref*{def:#1}}} %definition
\newcommand{\Alg}[1]{\hyperref[alg:#1]{Alg.~\ref*{alg:#1}}} %algorithm
\newcommand{\Ex}[1]{\hyperref[ex:#1]{Ex.~\ref*{ex:#1}}} %example
\newcommand{\Clm}[1]{\hyperref[clm:#1]{Claim~\ref*{clm:#1}}} %claim
\newcommand{\Assum}[1]{\hyperref[assump:#1]{Assumption~\ref*{assump:#1}}} %assumption
\newcommand{\Step}[1]{\hyperref[step:#1]{Step~\ref*{step:#1}}} %step
\newcommand{\Cond}[1]{\hyperref[cond:#1]{Condition~\ref*{cond:#1}}} %condition
\newcommand{\Item}[1]{\hyperref[item:#1]{Item~\eqref*{cond:#1}}} %item

\def\poly{{\rm poly}}
\newcommand{\polylog}{\poly(\log n, 1/\eps)}

\newcommand{\wt}{{\rm wt}}

\newcommand{\mH}{H}
\newcommand{\mL}{L}

\newcommand{\mG}{\mathcal{G}}

\renewcommand{\Pr}{\mathrm{Pr}}
\def\nofigures{0}

\ifnum\nofigures=0
\usepackage{tikz}
\usepackage{pstricks,tikz}
\usetikzlibrary{decorations.markings}
\usetikzlibrary{patterns,snakes}
\pgfdeclarelayer{background}
\pgfsetlayers{background,main}

\usepackage{verbatim}
\usepackage{fullpage}
\usepackage{graphicx,amsfonts,amsmath,amssymb,epsfig,color,multicol,pstricks,tikz, pgf}
\usetikzlibrary{calc, intersections}
\usepackage{pgf,tikz}

\fi

\ifnum\withcolors=1
\newcommand{\ttcom}[1]{{\color{red}{#1}}}

\newcommand{\strikeout}[1]{{\color{gray}{#1}}}

\newcommand{\tout}[1]{}

\newcommand{\mnote}[1]{{\small\color{magenta}(*)}\marginpar{\tiny\bf
\begin{minipage}[t]{0.5in}
\raggedright#1
\end{minipage}}}
\newcommand{\Sesh}[1]{{\color{red} Sesh says: #1}}

\else

	\newcommand{\ttcom}[1]{#1}
	\newcommand{\tout}[1]{}
	
	\newcommand{\mnote}[1]{}
	\newcommand{\strikeout}[1]{}
	\newcommand{\Sesh}[1]{#1}

\fi

\newcommand{\old}[1]{}

\newcommand{\whU}{{U_{i}}}
\newcommand{\wmom}{{\widetilde{M}_s}}
\newcommand{\walpha}{{\widetilde{\alpha}}}

\renewcommand{\th}{^\textrm{th}}

\newcommand{\mom}{M_s}
\newcommand{\momts}{M_{2s-1}}
\newcommand{\momone}{M_{1}}
\newcommand{\avgmom}{\mu}
\newcommand{\gm}{\widehat{\mom}}

\newcommand{\ostar}{O^*}
\newcommand{\mainalg}{\textbf{Moment-estimator}}

\def\eqdef{~\triangleq~}

\def\eps{\varepsilon}
\def\poly{{\rm poly}}
\usepackage[pagebackref,letterpaper=true,colorlinks=true,pdfpagemode=none,urlcolor=blue,linkcolor=blue,citecolor=BrickRed,pdfstartview=FitH]{hyperref}

\author{Talya Eden\thanks{School of EE, Tel Aviv
		University. {\tt talyaa01@gmail.com} \newline
		{This research was partially supported
			by a grant from the Blavatnik fund. The author is grateful to the Azrieli Foundation for the award of an Azrieli Fellowship.}}
	\and
	Dana Ron\thanks{School of EE, Tel Aviv University. {\tt danaron@tau.ac.il}\newline
	{This research was partially supported
		by the Israel Science Foundation grant No. 671/13 and by a grant from the Blavatnik fund.}}
	\and
	C. Seshadhri\thanks{Department of Computer Science, University of
		California, Santa Cruz. {\tt sesh@ucsc.edu}}
}

\title{Sublinear Time Estimation of Degree Distribution Moments:\\
	The Degeneracy Connection \\
	\ifnum\conf=1
	{\small (Extended Abstract)}
	\else
	{\small (Full Version)}
	\fi
}
\date{}

\def\setHeader{0}

\ifnum\setHeader=1

\usepackage{xcolor}
\usepackage[switch,columnwise]{lineno}
\usepackage{lipsum} 
\usepackage{fancyhdr}
\fi

\begin{document}

	\maketitle
	
	\begin{abstract}
	
	We revisit the classic problem of estimating the degree distribution moments of an undirected graph. Consider an undirected graph $G=(V,E)$ with $n$ (non-isolated) vertices, and define (for $s > 0$) $\mu_s = \frac{1}{n}\cdot\sum_{v \in V} d^s_v$. Our aim is to estimate $\mu_s$ within a multiplicative error of $(1+\varepsilon)$ (for a given approximation parameter $\varepsilon>0$) in sublinear time. We consider the sparse graph model that allows access to: uniform random vertices, queries for the degree of any vertex, and queries for a neighbor of any vertex. 
For the case of $s=1$ (the average degree), $\widetilde{O}(\sqrt{n})$ queries suffice for any constant
$\varepsilon$ (Feige, SICOMP 06 and Goldreich-Ron, RSA 08).  
Gonen-Ron-Shavitt (SIDMA 11) extended this result to all integral $s > 0$, by designing an algorithms that performs $\widetilde{O}(n^{1-1/(s+1)})$ queries. (Strictly speaking, their algorithm approximates the number of star-subgraphs of a given size, but a slight modification gives an algorithm for moments.)
	
	We design a new, significantly simpler algorithm for this problem. In the worst-case, it exactly matches the bounds of Gonen-Ron-Shavitt, and has a much simpler proof. More importantly, the running time of this algorithm is connected to the \emph{degeneracy} of $G$. This
is (essentially) the maximum density of an induced subgraph.
   For the family of graphs with degeneracy at most $\alpha$, it has a query complexity of $\widetilde{O}\left(\frac{n^{1-1/s}}{\mu^{1/s}_s} \Big(\alpha^{1/s} + \min\{\alpha,\mu^{1/s}_s\}\Big)\right) = \widetilde{O}(n^{1-1/s}\alpha/\mu^{1/s}_s)$. Thus, for the class of bounded degeneracy graphs (which includes all minor closed families and preferential attachment graphs), we can estimate the average degree in $\widetilde{O}(1)$ queries, and can estimate the variance of the degree distribution in $\widetilde{O}(\sqrt{n})$ queries. This is a major improvement over the previous worst-case bounds.
Our key insight is in designing an estimator for $\mu_s$
 that has low variance when $G$ does not have large dense subgraphs.
 \end{abstract}

\section{Introduction} \label{sec:intro}

Estimating the mean and moments of a sequence of $n$ integers $d_1, d_2, \ldots, d_n$ is a classic
problem in statistics that requires little introduction. In the absence of any knowledge
of the moments of the sequence, it is not possible to prove anything non-trivial.
But suppose these integers formed the degree sequence of a graph.
 % without isolates ({isolated vertices}).
Formally, let $G = (V,E)$ be an undirected %connected
graph
{over $n$ vertices, % and $m$ edges,
and let $d_v$ denote the degree of vertex $v\in V$, where we assume that $d_v\geq 1$
for every $v$}.\footnote{The assumption on there being no isolated vertices is made here only
for the sake of simplicity of the presentation, as it ensures a basic lower bound on the moments.}
% and let $d_i$ be the degree of the $i$th vertex.
Feige proved that $\ostar(\sqrt{n})$ uniform random vertex degrees (in expectation)
suffice to provide a
$(2+\eps)$-approximation
to the average degree~\cite{feige2006sums}.
(We use $\ostar(\cdot)$ to suppress $\poly(\log n, 1/\eps)$ factors.)
The variance can be as large as $n$ for graphs of constant average degree (simply consider a star),
but the constraints of a degree distribution allow for non-trivial approximations.
Classic theorems of Erd\H{o}s-Gallai and Havel-Hakimi characterize such sequences~\cite{Ha55,ErGa60,Ha62}.

Again, the star graph shows that the $(2+\eps)$-approximation cannot be beaten in sublinear time through
pure vertex sampling. % {(and degree queries)}.
Suppose we could also access random neighbors of a given vertex. In this setting, Goldreich and Ron showed it is possible to obtain a $(1+\eps)$-approximation to the average degree
in $\ostar(\sqrt{n})$ expected time~\cite{GR08}.

In a substantial % and %highly sophisticated
(and complex) generalization,
 Gonen, Ron, and Shavitt (henceforth, GRS) gave a sublinear-time algorithm
that estimates the higher moments of the degree distribution~\cite{GRS11}.
Technically, GRS gave an algorithm for approximating the number of stars in a graph, but a simple modification yields an algorithm for moment{s} estimation.
For precision, let us formally define this problem.
The \emph{degree distribution} is the distribution over the degree of a uniform random vertex. The {\em $s$-th moment} of the degree distribution is
$\avgmom_s \eqdef \frac{1}{n}\cdot \sum_{v\in V} d^s_v$.

%\begin{problem}[{\bf Degree Distribution Moment Estimation (DDME)}]
%\textbf{[Degree distribution moment estimation]}
%Consider an input undirected graph $G = (V,E)$.
%The number of vertices $n$ is known.
\ifnum\conf=1
\medskip
\else
\bigskip
\fi
\noindent{\bf The Degree Distribution Moment Estimation (DDME) Problem}.~
{\em
Let $G=(V,E)$ be a graph over $n$ vertices, where $n$ is known.
Access to $G$ is provided through the following queries.
We can
% (i) {sample a uniform random vertex,}
(i) get the {id (label)} of a uniform random vertex,
(ii) query the degree $d_v$ of any vertex $v$, (iii) query a uniform random neighbor
%\footnote{{}}
% D: In the future maybe add a comment about the ``standard'' neighbor-query model
%%%
of any vertex $v$.
Given $\eps > 0$ and $s \geq 1$, output a $(1+\eps)$-multiplicative approximation to $\avgmom_s$
with probability\footnote{The constant $2/3$ is a matter of convenience. % A simple boosting
It can be increased to at least $1-\delta$ by taking the median value of $\log(1/\delta)$
independent invocations.}  $> 2/3$.
%\end{problem}
}

\ifnum\conf=1
\medskip
\else
\bigskip
\fi

\sloppy
The DDME problem has important connections to network science, which is the study of
properties of real-world graphs. There have been numerous results on the significance
of heavy-tailed/power-law degree distributions in such graphs,
since the seminal results of Barab\'{a}si-Albert~\cite{BarabasiAlbert99,BrKu+00,FFF99}.
The degree distribution and its moments are commonly used
to characterize and model graphs appearing in varied applications~\cite{BiFaKo01,PeFlLa+02,ClShNe09,SaCaWiZa10,BiCh+11}.
On the theoretical side, recent results provide faster algorithms for graphs where
the degree distribution has some specified form~\cite{BeFoNo14,BrCy+16}.
Practical algorithms for specific cases of DDME have been studied by Dasgupta et al
and Chierichetti et al.~\cite{DaKu14,ChDa+16}. (These results requires bounds on the mixing time
of the random walk on $G$.)

\subsection{Results}

Let $m$ denote the number of edges in the graph (where $m$ is not provided to the algorithm). For the sake of simplicity, we restrict  the discussion in the introduction to
case when $\avgmom_s \leq n^{s-1}$. As observed by GRS,
the complexity of the DDME problem is %easier
{smaller} when $\avgmom_s$ is significantly larger.
GRS designed an {(expected)} $\ostar\left(n^{1-1/(s+1)}/\avgmom^{1/(s+1)}_s + n^{1-1/s}\right)$-query algorithm for DDME and proved this expression was optimal up to $\poly(\log n,1/\eps)$ dependencies.
{(Here $\ostar(\cdot)$ also suppresses additional factors that depend only on $s$).}
%, where in the case of GRS, there is an exponential dependence on $s$).}
 Note that for a graph without {isolated vertices}, $\avgmom_s \geq 1$ for every $s>0$, so this yields a worst-case $\ostar(n^{1-1/(s+1)})$ bound.
 % for graphs.
 The $s=1$ case is estimating the average degree,
so this recovers the $\ostar(\sqrt{n})$ bounds of Goldreich-Ron.
We mention a recent result by
Aliakbarpour et al.~\cite{ABGPRY16} for DDME, in a stronger model that assumes additional access to uniform random edges.
They get
a \old{significantly} better bound of
$\ostar(m/(n\avgmom_s)^{{1/s}})$ { in this stronger model, {for $s> 1$
(and $\avgmom_s \leq n^{s-1}$)}}.
\Sesh{Note that the main challenge of DDME 
is in measuring the contribution
of high-degree vertices, which becomes substantially easier
when random edges are provided. In the DDME problem without
such samples, it is quite non-trivial to even detect high degree
vertices.} 

All the bounds given above are known to be optimal, up to $\poly(\log n,1/\eps)$ dependencies,
and at first blush, this problem appears to be solved.
We unearth a connection between
DDME and the \emph{degeneracy} of $G$. 
\old{The \emph{arboricity}
of $G$ is the {minimum number of  forests into which its edges can be partitioned.}}
\Sesh{The degeneracy of $G$ is (up to a factor $2$) the maximum
density over all subgraphs of $G$.}
%maximum number of forests that the edges can be partitioned into.
We design an algorithm that has a nuanced query complexity, depending on the degeneracy of $G$.
% \ttcom{We assume the algorithm is given an upper bound on the arboricity, and further disccuss this assumption below.}
Our result subsumes all existing results, and provides substantial improvements in many interesting cases.
{Furthermore, our algorithm and its analysis are significantly simpler and more concise than in
the GRS result.}

We begin with a convenient corollary of our main theorem. A
%  messier bound appears as Theorem~\ref{thm:messy}. {Cite main theorem}
tighter,  {more precise bound appears as Theorem~\ref{thm:correctness}.}

\begin{theorem} \label{thm:bounded} Consider the family of graphs with degeneracy at most $\alpha$.
The DDME problem can be solved on this family using
% {with expected running time}
\ifnum\conf=0
\begin{eqnarray*}
\ostar\left(\frac{n^{1-1/s}}{\avgmom^{1/s}_s} \Big(\alpha^{1/s} + \min\{\alpha,\avgmom^{1/s}_s\}\Big)\right)\;
\end{eqnarray*}
\else
$\ostar\!\left(\frac{n^{1-1/s}}{\avgmom^{1/s}_s} \Big(\alpha^{1/s} + \min\{\alpha,\avgmom^{1/s}_s\}\Big)\right)$
\fi
queries {in expectation}.
{The running time is linear in the number of queries.}
\end{theorem}

Consider the case of \emph{bounded degeneracy} graphs, where $\alpha = O(1)$.
This is a rich class of graphs. \emph{Every} minor-closed family of graphs has bounded degeneracy, {as do} graphs generated by the Barab\'{a}si-Albert preferential attachment process~\cite{BarabasiAlbert99}.
\Sesh{There is a rich theory of \emph{bounded expansion graphs},
which spans logic, graph minor theory, and fixed-parameter
tractability~\cite{NeOs-book}. All these graph classes
have bounded degeneracy.}
For every such class of graphs, we get a $(1+\eps)$-estimate of $\avgmom_s$
in $\ostar(n^{1-1/s}/\avgmom^{1/s}_s)$ time.
\Sesh{We stress that bounded degeneracy does not imply any bounds
on the maximum degree or the moments. The star graph has degeneracy $1$, 
but has extremely large moments due to the central vertex.}

Consider any bounded degeneracy graph without isolated vertices.
We can accurately estimate the average degree $(s=1)$ in $\poly(\log n)$ queries,
and estimate the variance of the degree distribution $(s=2)$ in $\sqrt{n}\cdot\poly(\log n)$ queries.
Contrast this
with the (worst-case optimal) $\sqrt{n}$ bounds of Feige and Goldreich-Ron for average degree,
and the $\ostar(n^{2/3})$ bound of GRS for variance estimation.
For general $s$, our bound is a significant improvement over the
$\ostar(n^{1-1/(s+1)}/\avgmom^{1/(s+1)}_s)$
bound of GRS.

The algorithm attaining Theorem~\ref{thm:bounded} requires an upper bound on the degeneracy of the graph.
When an degeneracy bound is not given, the algorithm recovers
the bounds of GRS, with an improvement on the extra $\poly(\log n)/\eps$ factors.
More details are in Theorem~\ref{thm:correctness}.
% We note that the GRS bound cannot be improved, if no assumptions are made on the graph (except for
% the number of vertices). 
We note that the degeneracy-dependent bound in Theorem~\ref{thm:bounded} cannot be attained by an algorithm that is only given $n$ as a parameter.
% More precisely, 
% if an algorithm is only provided with $n$ and must work on all graphs with $n$ vertices, then 
% it must perform $\Omega(\ostar(n^{1-1/(s+1)}/\avgmom^{1/(s+1)}_s))$ queries
% even for graphs of constant arboricity.
In particular, if an algorithm is only provided with $n$ and must work on all 
graphs with $n$ vertices, then
it must perform $\Omega(\sqrt{n})$ queries in order to approximate the average degree
even for graphs of constant degeneracy (and constant average degree).
Details are given in Subsection 7.1 in the full version of the paper.

% without any further assumptions on the graph (except for the number of vertices) no algorithm that is required to succeed on all graphs can improve on the GRS bound.   }

%When the arboricity of $G$ is $\Theta(\alpha)$,  %we can show that
%the expression in Theorem~\ref{thm:bounded}
%is at most $\ostar(n^{1-1/(s+1)}/\avgmom^{1/(s+1)}_s + n^{1-1/s})$ (the query complexity of GRS).

%
% the query complexity bounds of GRS. The factor hidden by the $\ostar(\cdot)$ is $O(2^s \eps^{-2}\log n\log\log n)$, which is significantly better than
%
% {In general, we can show that the
% % query complexity stated in Theorem~\ref{thm:bounded} is upper bounded by the query complexity
% query complexity {(and running time) of our algorithm are upper bounded by those} of GRS.}

The bound of Theorem~\ref{thm:bounded} may appear artificial, but we
prove that it is optimal when $\avgmom_s \leq n^{s-1}$.
(For the general case,
we also have optimal upper and lower bounds.)
This construction is an extension of the lower bound proof of GRS.

\begin{theorem} \label{thm:lb-intro} Consider the family of graphs with degeneracy
% \mnote{D: removed ``at most''?}
$\alpha$ and where $\avgmom_s \leq n^{s-1}$.
Any algorithm for the DDME problem on this family requires
$\Omega\Big( \frac{n^{1-1/s}}{\avgmom^{1/s}_s}\cdot\Big(\alpha^{1/s} + \min\{\alpha,\avgmom^{1/s}_s\}\Big)\Big)$ queries.
\end{theorem}

% \begin{theorem} \label{thm:main} Consider the family of graphs with arboricity at most $\alpha$.
% The DDME problem can solved using the following number of queries:
% \begin{eqnarray*}
% \ostar\left(n^{1-1/s}\cdot\frac{\alpha^{1/s}}{\avgmom^{1/s}_s} + \min\Big\{ n^{1-1/s} \cdot \big(\frac{n^{s-1}}{\avgmom_s}\big)^{1-1/s} + \alpha\cdot\frac{n^{s-1}}{\avgmom_s}\Big\}\right)
%  \ \ \ \textrm{if $\avgmom_s > n^{s-1}$} \\
% \end{eqnarray*}
% \end{theorem}

\vspace{-3ex}
\subsection{From degeneracy to moment estimation} \label{subsec:arb-intro}

We begin with a closer look at the lower bound examples of Feige, Goldreich-Ron,
and GRS. The core idea is quite simple: DDME is hard when the overall graph is sparse,
but there are small dense subgraphs. Consider the case of a clique of size $100\sqrt{n}$
connected to a tree of size $n$. The small clique dominates the average degree, but
any sublinear algorithm with access only to  random vertices pays $\Omega(\sqrt{n})$
for a non-trivial approximation. GRS use more complex constructions to get an $\Omega(n^{1-1/(s+1)})$
lower bound for general $s$. This also involves embedding small dense subgraphs
that dominate the moments.

Can we prove a converse to these lower bound constructions? In other words, prove
that the \emph{non-existence} of dense subgraphs must imply that DDME is easier?
A convenient parameter for this non-existence is the \emph{degeneracy}.

%This is closely related to the degeneracy of a graph,
% and a classic theorem of
% Nash-Williams shows that the degeneracy is  basically the
% maximum average degree of a subgraph {\cite{nash1961edge, nash1964decomposition}}.
% 
But the degeneracy is a global parameter, and it is not clear how a sublinear algorithm
{can} exploit it. Furthermore, DDME algorithms are typically very local; they sample
random vertices, query {the degrees of these vertices} and maybe {also query}
 {the} degrees of some
{of {their}} neighbors. We need a local
property that sublinear algorithms can exploit, but can also be linked to {the} degeneracy.
We achieve this connection via the \emph{degree ordering} of $G$. Consider the DAG
obtained by directing all edges from lower to higher degree {vertices}. Chiba-Nishizeki related
the properties of the \emph{out-degree} distribution to the degeneracy, and exploited this for clique
counting~\cite{ChNi85}. Nonetheless, there is no clear link to DDME. (Nor do we use
any of their techniques; we state this result merely to show what led us to use the degree ordering).

Our main insight is the construction of an estimator for DDME whose variance depends
on the degeneracy of $G$. This estimator critically uses the degree ordering. Our proof
relates the variance of this estimator to the density of subgraphs in $G$, which can be bounded
by the degeneracy.
We stress that our algorithm is quite simple, and the technicalities are in the analysis and setting of certain parameters.
% Along the way, we get an alternate and simpler proof of the GRS bound.
% D: the above statement now appears earlier.

\vspace{-1ex}
\subsection{Designing the algorithm} \label{subsec:simple}

Designate the \emph{weight} of an edge $(u,v)$
to be $d^{s-1}_u + d^{s-1}_v$. A simple calculation yields that the sum of the weights
{of all edges}
is exactly $\mom \eqdef \sum_v d^s_v = n\cdot \avgmom_s$. Suppose we could sample uniform random \emph{edges} (and
knew the total number of edges).
Then we could hope to estimate $\mom$ through uniform edge sampling.
The variance of {the} edge weights can be bounded, and this yields an $\ostar(m/(n\avgmom_s)^{1/s}) = \ostar(n^{1-1/s})$ algorithm
(when no vertex is isolated).
Indeed, this is %almost identical
{very similar} to the approach of Aliakbarpour et al.~\cite{ABGPRY16}. Such variance
calculations were {also} used in the classic Alon-Matias-Szegedy result of frequency
moment estimation~\cite{AlonMS99}.

Our approach is to simulate uniform edge samples using uniform vertex samples.
Suppose we sampled a set $R$ of uniform random vertices. By querying the degrees
of all these vertices, {we can select vertices in $R$ with probability
proportional to their degrees, which allows us to uniformly sample edges that are incident
to vertices in $R$}.
% we can set up a data structure that samples
% a uniform random edge \emph{incident to $R$}.
Now, we simply run the uniform edge
sampling algorithm on these edges. This algorithmic structure was recently used for
sublinear triangle counting algorithms by Eden et al.~\cite{ELRS}.

%This leads to
{Here lies} the core technical challenge.
How to bound the number
of random vertices %required
{that is sufficient} for effectively simulating the random edge algorithm?
This boils down to the behavior of the variance of the ``vertex weight"
distribution.
Let the weight of a vertex be the sum of weights of its incident edges. The weight
distribution over vertices can be extremely skewed, and this approach would require
a forbiddingly large $R$.

A standard technique from triangle counting (first introduced
by Chiba-Nishizeki~\cite{ChNi85}) helps reduce the variance.
Direct all edges from lower degree to higher degree {vertices}, breaking
ties consistently. Now, set the weight of a vertex to be the sum of weights
on  incident \emph{out-edges}. Thus, a high-degree vertex with lower degree neighbors
will have a significantly reduced weight, reducing overall variance.
In the general case (ignoring degeneracy), a relatively simple argument bounds
the maximum weight of a vertex, {which enables us  to} bound the variance
of the weight distribution. This yields a much simpler algorithm and proof of the GRS bound.

{In the case of graphs with bounded degeneracy, we need a more refined approach.}
%Unfortunately, the variance %now
%\mnote{T: I'm not sure about this sentence, but don't know how to modify it.}
%has a %fairly complicated
%{non-trivial} form, and it is challenging
%to bound it without using the maximum weight.
Our key insight is an intimate connection between the variance and
{the existence  of} dense subgraphs in $G$.
We basically show that the main {structure} that leads to high variance
is the existence of dense subgraphs. Formally, we can translate a small upper bound on the density of any subgraph
to a bound on the variance of the vertex weights.
%Thus, our analysis actually shows that the existing lower bounds constructions are the ``only
%obstructions" for DDME.
\old{Finally, a classic theorem of Nash-Williams relates the arboricity to the maximum
density over subgraphs, and this allows for the proof of Theorem~\ref{thm:bounded}.}
\Sesh{This establishes the connection to the graph degenearcy.}
%
%
% Unfortunately, the total weight might be concentrated on edges incident to a few high degree vertices,
% and it is unlikely that a small $R$ can hit these vertices.
%
% All of this is really a prelude to the core technical challenge.
%

\vspace{-1ex}
\subsection{Simplicity of our algorithm} \label{subsec:alg}

Our viewpoint on DDME is quite different from GRS and its precursor~\cite{GR08},
which proceed by bucketing the vertices based on {their degree}. This leads to a complicated
algorithm, which essentially samples to estimate the size
of the buckets, and also {the number of edges between} various buckets {(and ``sub-buckets'')}.
%% New text
We make use of buckets in out analysis, in order to obtain the upper bound that
depends on the degeneracy $\alpha$ (in order to achieve the GRS upper bound, our analysis does
not use bucketing).

As explained above, our main DDME procedure, \mainalg{} is simple enough to present in a few lines
of pseudocode {(see Figure~\ref{fig:mainalg})}.  We feel that the structural simplicity of \mainalg{} is an important contribution of our work.

%\mainalg{} takes two sampling parameters $r$ and $q$.
%The full theorem can be proven using \mainalg, with an additional geometric search over
%the right parameters $r$ and $q$. In \mainalg\ we use $id(v)$ to denote the label (which is in $[n]$) of a vertex $v$.

{\mainalg{} takes two sampling parameters $r$ and $q$.
The main result Theorem~\ref{thm:correctness} follows from running \mainalg{}
with a {standard} geometric search for the right setting of $r$ and $q$.
In \mainalg\ we use $id(v)$ to denote the label %(which is in $[n]$)
of a vertex $v$, {where vertices have unique ids and there is a complete order over the ids}.
}
% For a set of vertices $R$, we use $E_R$ to denote the set of edges incident to $R$.

\begin{figure}[htb!]
\begin{center}
	\fbox{
		\begin{minipage}{0.95\textwidth}
			{\bf \mainalg$_s$}$(r,q)$
			\smallskip
			\begin{compactenum}
				\item Select $r$ vertices, uniformly, independently, at random and let the resulting multi-set be denoted by $R$. {Query the degree of each vertex in $R$, and let
    $d_R = \sum_{v\in R} d_v$.}\label{step:set_r}
%				\item Set up a data structure to allow sampling uniform edges in $E_R$.\label{step:ds-ER}
				\item For $i=1, \ldots, q$ do: \label{step:set_q}
			        \begin{compactenum}
					\item {Select a vertex $v_i$ with probability proportional to its degree
(i.e., with probability $d_{v_i}/d_R$), and query for a random neighbor $u_i$ of $v_i$.}
% and select $j$ uniformly at random from $\{1,\dots,d_v\}$.
%                    \item {Perform a neighbor query to obtain the $j^{\rm th}$ neighbor of
%                     $v_i$, and denote the answer by $u_i$.
%                    (Thus $(v_i,u_i)$ is a uniformly selected edge in $E_R$.)}
% Sample  an edge $(v_i,u_i)$ from $E_R$ uniformly, independently, at random, where $v_i \in R$.
					\item If  %$v_i \prec u_i$,
       $d_{v_i} < d_{u_i}$ or  $d_{v_i} = d_{u_i}$ and
                    $id(v_i) < id(u_i)$,
          set $X_i = (d^{s-1}_{v_i} + d^{s-1}_{u_i})$.
            Else, set $X_i = 0$. \label{step:def_Xi}
				\end{compactenum}
				\item Return $X=\frac{1}{r} \cdot \frac{d_R}{q} \cdot \sum\limits_{i=1}^q X_i \;.$ \label{step:def_X}
			\end{compactenum}
			\end{minipage}
	}
	\caption{Algorithm \mainalg$_s$ for approximating $\avgmom_s$.}
	\label{fig:mainalg}
\end{center}
\end{figure}
%\mnote{T: The algorithm requires access to the id of a vertex. Should I change the allowed queries to as it was before? D: I changed back (though still am not completely happy)}

%%%%%%%%%%%%%%%%%%%%%%%%%%%%%%%%
\iffalse
\begin{figure}[htb!]
	\fbox{
		\begin{minipage}{0.9\textwidth}
			{\bf \mainalg$_s$}$(r,q)$ \label{approx}
			\smallskip
			\begin{compactenum}
				\item Select $r$ vertices, uniformly, independently, at random and let the resulting momlti-set be denoted by $R$. \label{step:set_r}
				\item Set up a data structure to allow sampling uniform edges in $E_R$. (This is done by sampling $v \in R$ proportional to $d_v$, and then sampling a u.a.r. neighbor of $v$.)
				\item For $i=1, \ldots, q$ do: \label{step:set_q}
				\begin{compactenum}
					\item Sample  an edge $(v_i,u_i)$ from $E_R$ uniformly, independently, at random, where $v_i \in R$.
					\item If $d_{v_i} < d_{u_i}$ or (if $d_{v_i} = d_{u_i}$ and
                    $id(v_i) < id(u_i)$), set $X_i = (d^{s-1}_{v_i} + d^{s-1}_{u_i})$. Else, set $X_i = 0$. \label{step:def_Xi}
				\end{compactenum}
				\item Return $X=\frac{|E_R|}{rq} \cdot \sum\limits_{i=1}^q X_i \;.$ \label{step:def_X}
			\end{compactenum}
			
		\end{minipage}
	}
	\caption{Algorithm \mainalg$_s$ for approximating $\avgmom_s$}
	\label{fig:mainalg-intro}
\end{figure}
\fi
%%%%%%%%%%%%%%%%%%%%%%%%

\subsection{Other related work}\label{subsec:rel-work}
%DISCUSS THE ALGORITHM FOR ESTIMATING THE NUMBER OF STARS GIVEN RANDOM EDGES
As mentioned at the beginning of this section,
Aliakbarpour et al.~\cite{ABGPRY16} consider the problem of approximating the number
of $s$-stars for $s\geq 2$
when given access to uniformly selected edges. Given the ability to uniformly select edges, they
can select vertices with probability proportional to their degree (rather than uniformly).
This can be used to get an unbiased estimator of $\avgmom_s$ (or the $s$-star count)
with low variance. This leads to an {$O(m/(n\avgmom_s)^{1/s})$} bound, which is optimal {(for $\avgmom_s \leq n^{s-1}$)}.

Dasgupta, Kumar, and Sarlos give practical algorithms for average degree estimation,
though they assume bounds on the mixing time of the random walk on the graph~\cite{DaKu14}.
A recent paper of Chierichetti et al. build on these methods to sample nodes according to
powers of their degree (which is closely related to DDME)~\cite{ChDa+16}.
Simpson, Seshadhri, and McGregor give practical algorithms to estimate the entire cumulative degree distribution
in the streaming setting~\cite{SiSeMc15}. This is different from the sublinear query model
we consider, and the results are mostly empirical.

%
%
%
% If we
% select a vertex $v$ according to this distribution and query its degree $d_v$, then the expected
% value of the random variable $\frac{2m}{d_v}\cdot {d_v\choose s}$ is $\nu_s$.
% Since $d_v \leq \nu_s^{1/s}$ for every vertex $v$,
% the value of this random variable is upper bounded by $m\cdot \nu^{1-1/s}$, implying that
% it suffices to select $O\left(\frac{m}{\nu_s^{1/s}\eps^2}\right)$ vertices (edges).\footnote{The upper bound is actually a factor $s$ smaller, which we ignored for simplicity.} By~\cite{MPX07}, the number of edges, $m$ can be approximated to within $(1\pm \eps)$ given $\tilde{O}(\sqrt{n}/\eps^2)$ uniformly selected edges.
% Hence (using $m = O(\mom^{1/s}\cdot n^{1-1/s})$ for $m=\Omega(n)$), Aliakbarpour et al.~\cite{ABGPRY}  get an algorithm in the edge-sampling  model whose complexity is
% $\tilde{O}\left(\min\left\{\frac{n}{\nu_s^{1/s}},n^{1-1/s}\right\}/\eps^{2}\right)$. They
% also prove that this is tight for $\nu_s\leq n^s$.\mnote{D: what more would we like to say?}

In~\cite{ELRS}, Eden et al. present an algorithm for approximating the number of triangles in a graph.
Although this is a very different problem than DDME, there are similar challenges regarding high-degree vertices.
Indeed, as mentioned earlier, the approach of sampling random edges through a set of random vertices
was used in~\cite{ELRS}.

The degeneracy is closely related to other ``density" notions, such as the arboricity,
thickness, and strength of a graph~\cite{arb-wiki}. 
\Sesh{There is a rich history of
algorithmic results where run time depends on the degeneracy~\cite{MaBe83,ChNi85,AlGu08,EpLo10}.}
\old{There are numerous algorithmic results related to the degeneracy (which is at most twice the arboricity)~({cf.} \cite{MaBe83,ChNi85,AlGu08,EpLo10}).}

\ignore{Since in~\cite{GRS11} it was shown that degree and neighbor queries are not sufficient to get a sublinear-time
algorithm for this problem, Eden et al. also allow vertex-pair queries (that is, the algorithm may query whether there is an edge between a pair of vertices $u$ and $v$).
They give an algorithm whose expected query complexity is
	$\left(\frac{n}{t^{1/3}}+\min\left\{m,\frac{m^{3/2}}{t}\right\}\right)\cdot \polylog$,
and show that this result is tight (up to polylogarithmic factors in $n$ and the dependence on $1/\eps$).

What is common to the current work and the algorithm in~\cite{ELRS} is the high level structure of the algorithm.
The algorithm in~\cite{ELRS} also
uniformly selects a multi-set of vertices $R$ and then uniformly selects edges incident to vertices in $R$.
There too edges are assigned weights, where a weight of an edge is related to the number of
triangles it participates in. The algorithm and its analysis are more involved than in the current paper
due to several reasons. One is that the weight of an edge cannot be easily computed by performing
degree queries on its endpoints, and a procedure for approximating the weight is required. Another is that
the analysis of the variance of the random variables defined in the course of the algorithm is more subtle.
}
%There are {many}
{Other} sublinear algorithms for estimating various graph parameters include:
{approximating} the size of the minimum-weight spanning tree \cite{DBLP:journals/siamcomp/ChazelleRT05, DBLP:journals/siamcomp/CzumajS09, DBLP:journals/siamcomp/CzumajEFMNRS05}, maximum matching \cite{nguyen2008constant, yoshida2009improved} and of the minimum vertex cover \cite{DBLP:journals/tcs/ParnasR07,nguyen2008constant,DBLP:journals/talg/MarkoR09, yoshida2009improved, hassidim2009local, onak2012near}.

\ifnum\conf=1
\subsubsection*{A Comment regarding this extended abstract}
We defer some of the details of the analysis of the algorithm, as well as the lower bound proof,
to the accompanying full version of the paper.
\fi

	\section{The main theorem}

\begin{theorem} \label{thm:correctness}
	For every graph $G$, there exists an algorithm that returns a value
	$Z$ such that $Z \in [(1-\eps)\avgmom_s(G), (1+\eps)\avgmom_s(G)]$ with probability at least $2/3$.
	Assume that algorithm is given $\alpha$, an upper bound on the degeneracy of $G$.
	(If no such bound is provided, the algorithm assumes a trivial bound of $\alpha = \infty$.)
	The expected running time is the minimum of the following two expressions.
	
	\ifnum\conf=1
	\vspace{-1ex}
	\fi
	\begin{eqnarray}
	O\Big(2^s\cdot n^{1-1/s}\cdot\log^2 n\cdot \Big(\frac{\alpha}{\avgmom_s}\Big)^{1/s} +\min\Big\{\frac{n^{1-1/s}\cdot \alpha}{\avgmom_s^{1/s}}, \frac{n^{s-1}\cdot  \alpha}{\avgmom_s}\Big\}\Big)\cdot \;\frac{s \log n\cdot \log(s\log n)}{\eps^2} \label{eq:arb-rt}
	\end{eqnarray}
	% \[
	% \;\;\;\;\;\;\;\;\;=\;\ostar\Bigg(\min\Bigg\{\frac{n^{1-1/s}\cdot \alpha^{1/s}}{\avgmom_s^{1/s}},
	%                   \frac{n^{1-1/(s+1)}}{\avgmom_s^{1/(s+1)}}\Bigg\} +
	%      \min\Bigg\{\frac{n^{1-1/s}\cdot \alpha}{\avgmom_s^{1/s}}, \frac{n^{s-1}\cdot  \alpha}{\avgmom_s},n^{1-1/s},
	%  \frac{n^{s-1-1/s}}{\avgmom_s^{1-1/s}}\Bigg\}\Bigg) \;.
	% \]
	% \mnote{D: add a more friendly expression using $O_s$ and $\poly\log(n)$.}
	
	\ifnum\conf=1
	\vspace{-1ex}
	\fi
	\begin{eqnarray}
	% \lefteqn{O\left(\frac{n}{\mom^{1/(s+1)}}
	% +  \min\left\{n^{1-1/s}, \frac{n^{s-1/s}}{\mom^{1-1/s}}\right\}\right)\cdot\frac{s\log n\cdot \log(s\log n)}{\eps^2}}\\
	O\Big(\frac{n^{1-1/(s+1)}}{\avgmom_s^{1/(s+1)}} +
	\min\Big\{n^{1-1/s}, \frac{n^{s-1-1/s}}{\avgmom_s^{1-1/s}}\Big\}\Big)\cdot\frac{s\log n\cdot \log(s\log n)}{\eps^2} \label{eq:grs-rt}
	\end{eqnarray}
\end{theorem}

Equation~\eqref{eq:grs-rt} is essentially the query complexity of GRS (albeit with
a better dependence on $s, \log n$, and $1/\eps$). Thus, our algorithm is guaranteed to be at least
as good as that. If $\alpha$ is exactly the degeneracy of $G$, then we can prove that Equation~\eqref{eq:arb-rt}
is less than Equation~\eqref{eq:grs-rt}.
Within each expression, there is a min of two terms. The first term is smaller iff $\avgmom_s \leq n^{s-1}$.

The mechanism of deriving this rather cumbersome running time is the following. The algorithm of~Theorem~\ref{thm:correctness} runs \mainalg{} for geometrically %decreasing
increasing values of $r$ and $q$, which is in turn
derived from a geometrically decreasing guess of $\avgmom_s$.
It uses this guess to set $r$ and $q$. There is a setting of values depending on $\alpha$,
and a setting independent of it. The algorithm simply picks the minimum of these settings
to achieve the smaller running time.

\section{Sufficient conditions for $r$ and $q$ in \mainalg} \label{sec:alg}

In this section we provide sufficient conditions on the parameters $r$ and $q$
that are used by \mainalg \;(Figure~\ref{fig:mainalg}), in order for the algorithm to return a $(1+\eps)$ estimate of $\avgmom_s$.
%as presented in \Fig{mainalg}.
First we introduce some notations.
For a graph $G = (V,E)$ % over $n$ vertices and $m$ edges
%(whose vertices have unique ids)
and a vertex $v\in V$,
let $\Gamma(v)$ denote the set of neighbors of $v$ in $G$ {(so that
	$d_v = |\Gamma(v)|$).} For any (multi-)set $R$ of vertices, let $E_R$ be the (multi-)set of edges
incident to the vertices in $R$.
We will think of the edges in $E_R$ as ordered pairs; thus $(v,u)$ is distinct from $(u,v)$,
and so $E_R \eqdef \{(v,u): v\in R, u \in \Gamma(v)\}$.
{Observe that $d_R$, as defined in Step~\ref{step:set_r} of \mainalg\ equals $|E_R|$.}
Let $\mom = \mom(G) \eqdef \sum_{v\in V} d_v^s$, {so that $\avgmom_s = \mom/n$}.
In the analysis of the algorithm, it  is convenient  to work
with $\mom$ instead of  $\avgmom_s$.

A critical aspect of our algorithm (and proof) is the \emph{degree ordering on vertices}.
Formally, we set $u \prec v$ if $d_u < d_v$ or, $d_u = d_v$ and $id(u) < id(v)$.
Given the degree ordering, we let
$\Gamma^+(v) \eqdef \{u\in \Gamma(v): v\prec u\}$,  $d^+_v \eqdef |\Gamma^+(v)|$, and $E^+ \eqdef \{(v,u): v \in V, u \in \Gamma^+(v)\}$.
Here and elsewhere, we use $\sum_v$ as a shorthand for $\sum_{v\in V}$.

\begin{definition}\label{def:wt_e}
	We define the {\sf weight} of an edge $e=(v,u)$ as follows:
	if $v \prec u$ define $\wt(e) \eqdef (d_v^{s-1} + d_u^{s-1})$. Otherwise, $\wt(e) \eqdef 0$.\\
	For a vertex $v \in V$,
	$\wt(v) \eqdef \sum\limits_{u\in \Gamma(v)} \wt((v,u)) \; = \sum\limits_{u\in \Gamma^+(v)} \wt((v,u))$, and for a (multi-)set of
	vertices $R$, $\wt(R) \eqdef \sum\limits_{v\in R} \wt(v)$.
\end{definition}
{Observe that given the above notations and definition, \mainalg\
	selects uniform edges from $E_R$ and sets each $X_i$
	(in Step~\ref{step:def_Xi}) to $\wt((v_i,u_i))$. }
\ifnum\conf=0
Based on Definition~\ref{def:wt_e}, we obtain the next two claims, where the first claim connects between
$\mom$ and the weights of vertices.
\else
The next two claims readily follow from Definition~\ref{def:wt_e} (and the description of the algorithm).
\fi
\begin{claim} \label{clm:sum_wts} $\sum_v \wt(v) =\mom$.
\end{claim}

%%%
\ifnum\conf=0
\begin{proof} By the definition of the weights:
	\begin{align}
	\sum\limits_{v \in V} wt(v)
	= \sum\limits_{(v,u) \in E^+}(d_v^{s-1} + d_u^{s-1})
	= \sum\limits_{\{v,u\} \in E}(d_v^{s-1} + d_u^{s-1})
	= \sum\limits_{v \in V } \sum\limits_{u \in \Gamma(v)} d_v^{s-1}
	= \sum\limits_{v \in V } d_v^{s} = \mom \;, \nonumber
	% \label{eq:sum-dv-du-mom}
	\end{align}
	and the claim is established.
\end{proof}
\fi

\begin{claim} \label{clm:exp} $\EX[X]  = {\avgmom_s}$, where $X$ is as defined in Step~\ref{step:def_X} of the algorithm.
\end{claim}

%%%
\ifnum\conf=0
\begin{proof}
	Recall that $\wt(R) \eqdef \sum_{v \in R} \wt(v)$. Note that $X_i$
	(as defined in Step~\ref{step:def_Xi} of the algorithm) 	is exactly $\wt((v_i,u_i))$.
	Conditioning on $R$,
	$$\EX[X_i | R] = \frac{1}{|E_R|}\cdot\sum_{v \in R} \sum_{u \in \Gamma^+_v} \wt((v,u)) = \frac{1}{|E_R|} \cdot \sum_{v \in R} \wt(v) = \frac{1}{|E_R|} \cdot \wt(R)\;.$$
	{By the definition of $X$ in the algorithm (see Step~\ref{step:def_X}),}
	$$\EX[X|R] = \frac{{1}}{r}\cdot\frac{|E_R|}{q}\cdot\sum_{i=1}^q \EX[X_i|R] = \frac{{1}}{r}\cdot \wt(R)\;.$$
	Now, let us remove the conditioning.
	%	By uniformity of sampling % and by \Claim{sum_wts},
	{Since $\wt(R) \eqdef \sum_{v\in R}\wt(v)$, by linearity of the expectation,}
	\ifnum\conf=1
	$\EX_R[\wt(R)] = (r/n)\sum_v \wt(v)$,
	\else
	$$\EX_R[\wt(R)] = \frac{r}{n}\cdot \sum\limits_{v \in V} \wt(v)\;,$$
	\fi
	and thus {(using Claim~\ref{clm:sum_wts})},
	\ifnum\conf=1
	$\EX[X] = \EX_R[\EX[X|R]] = {\frac{1}{n}\cdot }\sum_v \wt(v) = {\avgmom_s}$\;.
	\else
	$$\EX[X] = \EX_R[\EX[X|R]] = {\frac{1}{n}\cdot }\sum\limits_{v \in V} \wt(v) = {\avgmom_s}\;,$$
	which completes the proof.
	\fi
\end{proof}
%%%
\fi

\subsection{Conditions on the parameters $r$ and $q$.}
We next state two conditions on the parameters $r$ and $q$, which are used in the algorithm, and then establish several claims, based on the
conditions holding. The conditions are stated in terms of properties of the graph as well as the approximation parameter $\eps$ and a confidence parameter $\delta$.
% {It will be convenient for us to state the conditions on $r$ and $q$ in terms of
% $\mom$ (rather than $\avgmom_s$).}

\begin{enumerate}
	\item {\bf The vertex condition:}
	\ifnum\conf=1
	$r \geq (120 \cdot n \cdot \sum_v \wt(v)^2)/(\eps^2 \cdot \delta \cdot \mom^2)$,
	\else
	$$r \geq \frac{30\cdot n \cdot \sum_v \wt(v)^2}{\eps^2 \cdot \delta \cdot \mom^2}\;.$$
	\fi
	\label{cond:r}
	\item {\bf The edge condition:}
	\ifnum\conf=1
	$q \geq  2000 \cdot m \cdot \momts/(\eps^2\cdot \delta^3\cdot \mom^2)\;.$
	\else
	$$q \geq  \frac{2000 \cdot m \cdot \momts}{\eps^2\cdot \delta^3\cdot \mom^2}\;.$$
	\fi
	\label{cond:q}
\end{enumerate}

%The following boils down to two applications of Chebyshev's inequality, one for the vertices and one for the edges.

\begin{lemma} \label{lem:vert-samp}
	If Condition~\ref{cond:r} holds, then with probability at least $1-\delta/2$, all the following hold.
	%\ifnum\conf=1
	%	\begin{asparaitem}
	%		\item  $\wt(R) \in \left[\left(1-\frac{3\eps}{2}\right)
	%\cdot \frac{r}{n}\cdot \mom,\left(1+\frac{\eps}{2}\right)\cdot \frac{r}{n}\cdot\mom\right]$.
	%		\item $|E_R| \leq \frac{12}{\delta}\cdot\frac{r}{n}\cdot m$.
	%		\item  $\sum_{(v,u) \in E^+_R} \wt\left((v,u)\right)^2 \leq \frac{18}{\delta}\cdot %\frac{r}{n}\cdot \momts$.	\end{asparaitem}
	%\else
	\begin{enumerate}
		\item  $\wt(R) \in \left[\left(1-\frac{\eps}{2}\right)
		\cdot \frac{r}{n}\cdot \mom,\left(1+\frac{\eps}{2}\right)\cdot \frac{r}{n}\cdot\mom\right]$.
		\item $|E_R| \leq \frac{12}{\delta}\cdot\frac{r}{n}\cdot m$.
		\item  $\sum\limits_{(v,u) \in E^+_R} \wt\left((v,u)\right)^2 \leq \frac{18}{\delta}\cdot \frac{r}{n}\cdot \momts$.	
		% \mnote{D: should we use here $r\cdot \avgmom_{2s-1}$ instead of $\frac{r}{n}\cdot \momts$?}
	\end{enumerate}
	%\fi
\end{lemma}
\ifnum\conf=1
The proof of the first item in Lemma~\ref{lem:vert-samp} follows from Chebyshev's inequality (using
$\var[\wt(R)] \leq \frac{r}{n} \cdot\sum_v \wt(v)^2$), and the proofs of the
other two items follow from Markov's inequality (as well as the definition of $\momts$).
\else

\begin{proof}
	First, we look at the random variable $\wt(R)$. %Note that
	By the definition of $\wt(R)$ and Claim~\ref{clm:sum_wts},
	$\EX[\wt(R)] = (r/n)\cdot \sum_v \wt(v) = (r/n)\cdot \mom$.
	% \else
	% $$\EX[\wt(R)] = \frac{r}{n}\cdot \sum_v \wt(v) = \frac{r}{n}\cdot \mom\;.$$
	% \fi
	Turning to the variance of $\wt(R)$, since the vertices in $R$ are chosen uniformly at random,
	\begin{eqnarray*}
		\var_R[\wt(R)] = r\cdot \var_v[\wt(v)] = r\Bigg(\frac{1}{n} \cdot\sum_v \wt(v)^2 - \Big(\frac{1}{n}\cdot \sum_v \wt(v)\Big)^2\Bigg) \leq \frac{r}{n} \cdot\sum_v \wt(v)^2.
	\end{eqnarray*}
	By Chebyshev's inequality,
	%\ifnum\conf=1
	%	\[ \Pr\Big[\Big|\wt(R) - \EX[\wt(R)]\Big| \geq \frac{\eps}{2} \cdot \EX[\wt(R)]\Big] = %\frac{4\var[\wt(R)]}{\eps^2 \cdot \EX[\wt(R)]^2} \leq \frac{4 (r/n)\sum_v \wt(v)^2}{\eps^2 \cdot (r/n)^2 %\cdot \mom^2}
	%	= \frac{4 n \sum_v \wt(v)^2}{\eps^2 \cdot \mom^2 \cdot r} \;. \]
	%\else
	\begin{eqnarray*}
		\Pr\Big[\Big|\wt(R) - \EX[\wt(R)]\Big| \geq \frac{\eps}{2}\cdot \EX[\wt(R)]\Big]
		&\leq& \frac{4\var[\wt(R)]}{\eps^2 \cdot \EX[\wt(R)]^2} \\
		&=& \frac{4 (r/n)\sum_v \wt(v)^2}{\eps^2 \cdot (r/n)^2 \cdot \mom^2} %\\
		\;=\; \frac{4 n \sum_v \wt(v)^2}{\eps^2 \cdot \mom^2 \cdot r} \;.
	\end{eqnarray*}
	%\fi
	Applying the lower bound on $r$ from Condition~\ref{cond:r}, %(with an appropriate constant $c_r$),
	this probability is at most $\delta/6$. (Indeed, Condition~\ref{cond:r} was defined as such to
	get this bound.)
	
	The other bounds follow simply from  Markov's inequality. Observe that $\EX[|E_R|] = (r/n)(2m)$, and so
	$\Pr[|E_R| > (12/\delta)(r/n)m] \leq \delta/6$.
	
	The random variable $Y = \sum_{(v,u) \in E^+_R} \wt\left((v,u)\right)^2$
	% is non-negative, and % we will simply apply
	(which is non-negative), satisfies
	\begin{eqnarray}
	\EX[Y] & = & \frac{r}{n}\cdot  \sum_{(v,u) \in E^+} \wt\left((v,u)\right)^2 = \frac{r}{n}\cdot\sum_{(v,u) \in E^+} (d^{s-1}_v + d^{s-1}_u)^2 \nonumber \\
	& \leq & 3\cdot \frac{r}{n}\cdot \sum_{(v,u) \in E^+} (d^{2s-2}_v + d^{2s-2}_u) = 3 \cdot \frac{r}{n}\cdot\sum_v d^{2s-1}_v = 3 \cdot \frac{r}{n}\cdot\momts\;.
	\label{eq:EX-Y}
	\end{eqnarray}
	By Markov's inequality,
	$\Pr[Y \geq ({18}/\delta)(r/n)\momts] \leq \delta/6$. We apply a union bound to complete the proof.
\end{proof}
\fi
%%%%%%%%%%%%%%%%%%

\begin{theorem} \label{thm:cond} If Conditions~\ref{cond:r} and~\ref{cond:q} hold, then
	$X \in [(1-\eps){\avgmom_s}, (1+\eps){\avgmom_s}]$
	with probability at least $1-\delta$.
\end{theorem}

\begin{proof}
	Condition on any choice of $R$. We have $\EX[X|R] = ({1}/r) \wt(R)$. Turning to the variance, since the edges $(v_i, u_i)$ are chosen from $E_R$ uniformly at random,
	\ifnum\conf=0
	\begin{eqnarray*}
		\var[X|R] = \left(\frac{{1}}{r}\right)^2 \cdot \left(\frac{|E_R|}{q}\right)^2
		\cdot \var\left[\sum_{i=1}^q X_i\;\Big|\;R\right] & = & \left(\frac{{1}}{r}\right)^2
		\cdot \frac{|E_R|^2}{q}\cdot \var[X_1|R] \\
		& \leq &  \left(\frac{{1}}{r}\right)^2\cdot \frac{|E_R|^2}{q}\cdot \EX[(X_1|R)^2] \\
		& = & \left(\frac{{1}}{r}\right)^2\cdot \frac{|E_R|}{q}\cdot
		\sum_{(v,u) \in E^+_R} \wt((v,u))^2  \\
		& = & \frac{1}{q}\cdot\frac{ |E_R|}{r}\cdot \frac{\sum_{(v,u) \in E^+_R} \wt\left((v,u)\right)^2}{r} \;.
	\end{eqnarray*}
	\else
	it is not hard to verify that
	\[
	\var[X|R] = \left(\frac{{1}}{r}\right)^2 \cdot \left(\frac{|E_R|}{q}\right)^2
	\cdot \var\left[\sum_{i=1}^q X_i\;\Big|\;R\right]
	= \frac{1}{q}\cdot\frac{ |E_R|}{r}\cdot \frac{\sum_{(v,u) \in E^+_R} \wt\left((v,u)\right)^2}{r}\;.
	\]
	\fi
	Let us now condition on $R$ such that the bounds of Lemma~\ref{lem:vert-samp} hold.
	Note that such an $R$ is chosen with probability at least $1-\delta/2$.
	We get
	\ifnum\conf=0
	$$ \var[X|R] \leq \frac{250}{\delta^2} \cdot \frac{1}{q} \cdot  \frac{m }{{n}} \cdot
	\frac{\momts}{{n}} \;.$$
	\else
	$ \var[X|R] \leq \frac{250}{\delta^2} \cdot \frac{1}{q} \cdot  \frac{m }{{n}} \cdot
	\frac{\momts}{{n}}$.
	\fi
	We apply %(yet again!)
	Chebyshev's inequality %again
	and invoke Condition~\ref{cond:q}:
	% (with an appropriate constant   $c_q$).
	\begin{eqnarray*}
		\Pr\left[\Big|(X|R) - \EX[X|R]\Big| \leq \frac{\eps}{2}\cdot {\avgmom_s}\right] & \leq & \frac{4 \cdot \var[X|R]}{\eps^2 \cdot {\avgmom_s^2}} \leq  \frac{1}{q} \cdot \frac{4  \cdot (250/\delta^2)\cdot  m \cdot \momts}{\eps^2\cdot \mom^2} \leq \frac{\delta}{2} \;.
	\end{eqnarray*}
	By Lemma~\ref{lem:vert-samp}, $\EX[X|R] = ({1}/r)\wt(R) \in [(1-\eps/2){\avgmom_s},(1+\eps/2){\avgmom_s}]$.
	% Combining everything,
	\ifnum\conf=0
	By taking into account both the probability that $R$ does not satisfy one (or more) of the bounds in Lemma~\ref{lem:vert-samp}
	and the probability that $X$ (conditioned on $R$ satisfying these bounds) deviates
	by more than $(\eps/2) {\avgmom_s}$ from its expected value,
	we get  that $|X - {\avgmom_s}| < \eps {\avgmom_s}$ with probability at least $(1-\delta/2)^2 > 1-\delta$.
	\else
	The theorem follows by applying the union bound.
	\fi
\end{proof}

%%%
\ifnum\conf=0

\subsection{The Algorithm with Edge Samples}\label{sec:edge-samples}

% \mnote{Returned as subsection}
As an aside, %Theorem~\ref{thm:cond}
%{the analysis in \Sec{alg}}
{the above analysis can be
	%can be used
	slightly adapted} to prove the result of Aliakbarpour et al.~\cite{ABGPRY16}
% \cite{aliakbarpour2016sublinear}
on estimating
moments using random edge queries. Observe that we can then simply set $R = V$ and $r=n$ in \mainalg$_s$. {This immediately gives $\wt(R) = \mom$, $|E_R|=2m$, and $\sum_{(v,u) \in E^+_R} \wt((v,u))^2 \leq 3\momts$ (as shown in Equation~\eqref{eq:EX-Y}).
	Similarly to what is shown in the proof of Theorem~\ref{thm:cond},
	$\var[X] = O(q^{-1}\cdot m \cdot \momts)$ (where $X$ is as defined in Step~\ref{step:def_X} of the algorithm).}
% \mnote{D: I added some details since I didn't see how we get the result in a completely modular way.
% Maybe I am missing something.}
{As shown in Equation~\eqref{eq:bound_mu_2s-1},
	$\momts  \leq \mom^{2-1/s}$.}
% A simple norm inequality provides the bound on $q$.
% \[ \momts = \sum_v d^{2s-1}_v \leq \left(\sum_v d^s_v\right)^{(2s-1)/s} = \mom^{2-1/s} \;.
% \numberthis \label{bound_mu_2s-1} \]
Thus, if we set $q \geq \frac{c_q}{\eps^2\cdot\delta^3}\cdot\frac{m}{\mom^{1/s}}$ (for a sufficiently large constant $c_q$), we satisfy Condition~\ref{cond:q}.
This is exactly the bound of Aliakbarpour et al.
\fi
%%%

\section{Satisfying Conditions \ref{cond:r} and \ref{cond:q} in general graphs}  \label{sec:GRS}

We show how to set $r$ and $q$ to
satisfy Conditions~\ref{cond:r} and~\ref{cond:q} in general graphs.
Our setting of $r$ and $q$ will give us the same query complexity as~\cite{GRS11}
(up to the dependence on $1/\eps$ and $\log n$, on which we improve, and the exponential
dependence on $s$ in~\cite{GRS11}, which we do not incur).
% \mnote{D: should double check in terms of $\log n$. Holds for $1/\eps$.}
In the next section we show how %the query complexity can be significantly improved
the setting of $r$ and $q$ can be  improved using a degeneracy bound.
% \mnote{Justify why we give this section rather than jumping to bounded arboricity?}
% In what follows we prove that the following setting satisfies the conditions.
% We let

For $c_r$ and $c_q$ that are sufficiently large constants, we set
\begin{equation}
r =  % r(\eps,\delta,n,\mom)
\frac{c_r}{\eps^2\cdot\delta} \cdot \frac{n}{\mom^{1/(s+1)}}\;,
\ \ \ \ \ q =  %q(\eps,\delta,n,\mom)
\frac{c_q}{\eps^2\cdot\delta^3} \cdot \min\left\{n^{1-1/s}, \frac{n^{s-1/s}}{\mom^{1-1/s}}\right\} \;. \label{eq:cond}
\end{equation}
This setting of parameters requires the knowledge of $\mom$, which is exactly
what we are trying to approximate ({up to the normalization factor of $n$}). A simple geometric search argument
% shows that we can get %this query complexity
% {query complexity $O((r+q)\cdot\poly(\log n))$}
% \emph{without} knowledge of $\mom$.
{alleviates the need to know $\mom$.}
% For convenience, we defer this argument to
{For details see} Section~\ref{sec:search_no_alpha}.
%\Sec{search_no_alpha}.

%%%%%%%%
\iffalse
In order to show that $r$ as set in Equation~\eqref{eq:cond} satisfies Condition~\ref{cond:r}, we shall make use of
% simple but critical property of
% D: removed above since we don't exactly say what this property is.
the degree ordering. Let $\theta = \mom^{\frac{1}{s+1}}$ be a degree threshold. We define
\ifnum\conf=1
$\mH \eqdef \{v: d_v > \theta \}$ %(the set of ``high-degree'' vertices), and
$\mL \eqdef V\setminus H$.% (the set of ``low-degree'' vertices).
\else
$$\mH \eqdef \{v: d_v > \theta \} \;\;\;\mbox{ and }\;\;\; \mL \eqdef V\setminus H\;.$$
\fi
{This partition into ``high-degree'' vertices ($\mH$) and ``low-degree'' vertices ($\mL$) will
	be useful in upper bounding   $\sum_v \wt(v)^2$ (recall Condition~\ref{cond:r}) by
	upper bounding the maximum weight $\wt(v)$ of a vertex $v$.
}
\fi
%%%%%%%%%%%%%

%%%%	
\ifnum\conf=0
In what follows (and elsewhere)  we  make use of H\"older's inequality:
\begin{theorem}[H\"older's inequality]
	For values $p$ and $q$ such that $p,q>1$ and $\frac{1}{p}+\frac{1}{q}=1$,
	\ifnum\conf=0
	\[\sum\limits_{i=1}^k |x_i \cdot y_i| \leq \left(\sum_{i=1}^k |x_i|^p \right)^{\frac{1}{p}} \left(\sum_{i=1}^k |y_i|^q \right)^{\frac{1}{q}}. \]  %\numberthis \label{eqn:holder} \]
	\else
	$\sum\limits_{i=1}^k |x_i \cdot y_i| \leq \left(\sum_{i=1}^k |x_i|^p \right)^{\frac{1}{p}} \left(\sum_{i=1}^k |y_i|^q \right)^{\frac{1}{q}}$.
	\fi
	We refer to $p$ and $q$ as the {\sf conjugates} of the formula.
	\label{thm:holder}
\end{theorem}
\fi

In order to assert that $r$ as set in Equation~\eqref{eq:cond} satisfies Condition~\ref{cond:r}, it suffices
to establish the next lemma.
\begin{lemma}[Condition~\ref{cond:r} holds] \label{lem:cond-r}
	\ifnum\conf=0
	$$\sum_v \wt(v)^2 \leq 4\mom^{2-\frac{1}{s+1}} \;.$$  %= 4\mu^{1 + s/(s+1)}_s\;.$$
	\else
	$\sum_v \wt(v)^2 \leq 4\mom^{2-\frac{1}{s+1}} $.
	\fi
\end{lemma}

\begin{proof}
	Let $\theta = \mom^{1/(s+1)}$ be a degree threshold. We define
	\ifnum\conf=1
	$\mH \eqdef \{v: d_v > \theta \}$, %(the set of ``high-degree'' vertices), and
	$\mL \eqdef V\setminus H$. % (the set of ``low-degree'' vertices).
	\else
	$$\mH \eqdef \{v: d_v > \theta \} \;\;\;\mbox{ and }\;\;\; \mL \eqdef V\setminus H\;.$$
	\fi
	This partition into ``high-degree'' vertices ($\mH$) and ``low-degree'' vertices ($\mL$) will
	be useful in upper bounding  the maximum weight $\wt(v)$ of a vertex $v$, and hence
	upper bounding $\sum_v \wt(v)^2$. Details follow.
	
	We first observe that
	$|\mH| \leq \mom^{1/(s+1)}$. This is true since
	otherwise,
	$\sum_{v\in \mH} d_v^s > \mom^{1/(s+1)}\cdot \mom^{\frac{s}{s+1}} = \mom $,	which is a contradiction. 	 We claim that this upper bound on $|\mH|$ implies that
	\begin{equation}
	\max_v d^+_v \leq \mom^{1/(s+1)}\;.
	\label{eq:max-d-plus}
	\end{equation}
	To verify this, assume, contrary of the claim, that for some $v$, $d^+_v > \mom^{1/(s+1)}$. But then there are at least $\mom^{1/(s+1)}$ vertices $u$ such that $d_u \geq d_v \geq d^+_v > \mom^{1/(s+1)}$. This contradicts the bound on $|\mH|$.
	
	It will also be useful to bound $\sum_{u\in \mH} d_u^{s-1}$.
	By H\"older's inequality with conjugates $s$ and $s/(s-1)$
	\ifnum\conf=1
	(a statement of H\"older's inequality can be found in the full version of the paper)
	\fi
	and the bound on $|\mH|$,
	\begin{equation}
	\sum\limits_{u\in \mH} d_u^{s-1} =	\sum\limits_{u\in \mH} 1\cdot d_u^{s-1} \leq |\mH|^{1/s}\left(\sum \limits_{u \in \mH} d_u^s\right)^{\frac{s-1}{s}}  \leq \mom^{\frac{1}{s(s+1)}}\cdot\mom^{\frac{s-1}{s}} \leq  \mom^{\frac{s}{s+1}} \;.
	\label{eq:sum-H-du-s-1}
	\end{equation}
	
	We now turn to bounding $\max_v\{\wt(v)\}$.
	By the definition of $\wt(v)$ and the degree ordering,
	\begin{equation}
	\wt(v) = \sum_{u\in \Gamma^+(v)}(d_v^{s-1} + d_u^{s-1}) \leq 2\sum_{u \in \Gamma^+(v)} d^{s-1}_u
	= 2\sum_{u \in \Gamma^+(v) \cap \mL} d^{s-1}_u + 2\sum_{u \in \Gamma^+(v) \cap \mH} d^{s-1}_u \;.
	\label{eq:max-wt-v}
	\end{equation}
	For the first term on the right-hand-side of Equation~\eqref{eq:max-wt-v},
	recall that $d_u \leq \mom^{1/(s+1)}$ for $u \in \mL$.
	Thus, by Equation~\eqref{eq:max-d-plus},
	\begin{equation}
	\sum_{u \in \Gamma^+(v) \cap \mL} d^{s-1}_u \leq d^+_v \cdot \mom^{\frac{s-1}{s+1}} \leq \mom^{\frac{s}{s+1}} \;.
	\label{eq:sum-u-mL}
	\end{equation}
	For the second term, using $ \Gamma^+(v) \cap \mH \subseteq \mH$ and
	applying Equation~\eqref{eq:sum-H-du-s-1},
	\begin{equation}
	\sum_{u \in \Gamma^+(v) \cap \mH} d^{s-1}_u \leq \sum_{u \in \mH} d^{s-1}_u
	\leq \mom^{\frac{s}{s+1}}\;.
	\label{eq:sum-u-mH}
	\end{equation}
	Finally,
	$$\sum_v \wt(v)^2 \leq \max_v \{\wt(v)\}\cdot \sum_v \wt(v) \leq \mom^{2-1/(s+1)}\;,$$
	where  the second inequality follows by combining Equations~\eqref{eq:max-wt-v}--\eqref{eq:sum-u-mH}
	to get an upper bound on $\max_v\{\wt(v)\}$ and applying Claim~\ref{clm:sum_wts}.
\end{proof}

The next lemma implies that Condition~\ref{cond:q} holds for $q$ as set in Equation~\eqref{eq:cond}.
\begin{lemma}[Condition~\ref{cond:q} holds] \label{lem:cond-q}
	\ifnum\conf=0
	$$ \min\left\{n^{1-1/s}, \frac{n^{s-1/s}}{\mom^{1-1/s}}\right\} \geq 2m \cdot \frac{\momts}{\mom^2} \;.$$
	\else
	$ \min\left\{n^{1-1/s}, \frac{n^{s-1/s}}{\mom^{1-1/s}}\right\} \geq 2m \cdot \frac{\momts}{\mom^2}$.
	\fi
\end{lemma}

\begin{proof} We can bound $\momts$ in two ways.
	% \tcom{First, is by $\mom^{2-1/s}$, as shown in Equation~\eqref{bound_mu_2s-1}.}
	First, by a standard norm inequality, since $s \geq 1$,
	% \begin{equation}
	% \momts = \Bigg(\Big(\sum_v d^{2s-1}_v\Big)^{1/(2s-1)}\Bigg)^{2s-1} \leq \Bigg(\Big(\sum_v %
	% d^s_v\Big)^{1/s}\Bigg)^{2s-1} = \mom^{2-1/s}. \label{eq:bound_mu_2s-1}
	% \end{equation}
	\begin{align}
	 \momts = \sum_v d^{2s-1}_v \leq \left(\sum_v d^s_v\right)^{(2s-1)/s} = \mom^{2-1/s} \;. \label{eq:bound_mu_2s-1}
	 \end{align}
	We can also use the trivial bound $d_v \leq n$ and get $\momts \leq n^{s-1} \cdot\mom$.
	Thus, $\momts \leq \min \{\mom^{2-1/s}, n^{s-1}\cdot \mom\}$.
	%%	
	%	Note that $s$ and $s/(s-1)$ are H\"{o}lder conjugates.
	By applying H\"{o}lder's inequality with conjugates $s/(s-1)$  and $s$ we get that
	\begin{equation}
	2m = \sum_v 1\cdot d_v \leq
	%\left(\sum_v 1\right)^{(s-1)/s}
	% D: In the previous application we didn't go into this detail so not clear why here
	n^{(s-1)/s}\cdot \left(\sum_v d^s_v\right)^{1/s} = n^{1-1/s}\cdot \mom^{1/s}\;.
	\label{eq:m-n-mu}
	\end{equation}
	We multiply the bound by $\momts$ to complete the proof.
\end{proof}

\section{The Degeneracy Connection}\label{sec:arb}
% \mnote{Since preliminaries are short, not clear that makes sense to have separate section, so
% incorporated.}
\Sesh{The degeneracy, or the coloring number,
 of a graph $G = (V,E)$ is the maximum
value, over all subgraphs $G'$ of $G$, of the minimum degree in $G'$.
 In this definition, we can replace ``minimum" by ``average''
to get a $2$-factor approximation to the degeneracy
(refer to~\cite{wiki-degen}; Theorem 2.4.4 and Corollary 5.2.3 of~\cite{diestel}). Abusing notation, it will be convenient for us to define
%$\alpha(G) =\max_{S \subseteq V} \left\{ \left\lceil \frac{|E(S)|}{|S|-1} \right\rceil \right\}$.
$\alpha(G) =\max_{S \subseteq V} \left\{ \left\lceil \frac{|E(S)|}{|S|-1} \right\rceil \right\}$.
}
\old{Recall that the arboricity of  $G = (V,E)$, denoted $\alpha(G)$ is
the minimum number of forests into which $E$ can be partitioned.}
% As proved by Nash-Williams~\cite{nash1961edge,nash1964decomposition}, $\alpha(G)$ also equals
% $\max_{S\subseteq V} \left\{\left\lceil\frac{|E(S)|}{|S|-1}\right\rceil\right\}$, where $E(S)$ denotes
% the set of edges in the subgraph induced by $S$.
% \mnote{D: can move theorem from arboricity section here, or move def there}
\old{We shall make use of the next theorem, which  is due to Nash-Williams~\cite{nash1961edge,nash1964decomposition}.
\begin{theorem}[Nash-Williams' equality] \label{thm:Nash-Williams}
	%	Let $G$ be a graph with arboricity $\alpha(G)$. For a subset of vertices  $S\subseteq V$, let
	% $n_S$ and $m_S$ denote the number of vertices and edges in the subgraph induced by $S$, respectively.
	% $E(S)$ denote the set of edges in the subgraph induced by $S$. Then
	For every graph $G$,
	\ifnum\conf=0
	\[\alpha(G) = \max_{S \subseteq V} \left\{ \left\lceil \frac{|E(S)|}{|S|-1} \right\rceil \right\}\;, \]
	\else
	$\alpha(G) = \max_{S \subseteq V} \left\{ \left\lceil \frac{|E(S)|}{|S|-1} \right\rceil \right\}$,
	\fi
	where $E(S)$ denotes the set of edges in the subgraph induced by $S$.
\end{theorem}
}

We also make the following observation regarding the relation between $\alpha(G)$ and $\mom(G)$.
\begin{claim}\label{clm:alpha-mom}
	For every graph $G$,
	$\alpha(G) \leq \mom(G)^{\frac{1}{s+1}}$.
	%	$\alpha(G) \leq \mom(G)^{{1}/{(s+1)}}$.
\end{claim}
\ifnum\conf=0
\begin{proof}
	Let $S$ be a subset of vertices that maximizes $\left\lceil \frac{|E(S)|}{|S|-1} \right\rceil$,
	and let $\bar{d}(S)$ denote the average degree in the subgraph induced by $S$.
	Then $\bar{d}(S) = \frac{2|E(S)|}{|S|} \geq \left\lceil \frac{|E(S)|}{|S|-1}\right\rceil = \alpha(G)$.
	Hence, $|S| \geq \alpha(G)$, and by H\"older's inequality (Theorem~\ref{thm:holder}) with conjugates $s/(s-1)$ and $s$,
	\[\alpha(G)\cdot |S| \leq \sum\limits_{v \in S }d_v \leq \left(\sum\limits_{v \in S }d_v^s\right)^{\frac{1}{s}} \cdot |S|^{1-\frac{1}{s}},\] implying that
	$\alpha(G)\cdot |S|^{\frac{1}{s}} \leq \mom(G)^{\frac{1}{s}}$. Since $|S| \geq \alpha(G)$, we get that $\alpha(G) \leq \mom(G)^{\frac{1}{s+1}}$.
\end{proof}
\fi

\bigskip
In this section, we show that the following setting of parameters for \mainalg$_s$
satisfies Conditions~\ref{cond:r} and~\ref{cond:q}, for every graph $G$ with degeneracy at most $\alpha$ (i.e., $\alpha(G) \leq \alpha$),
and for appropriate constants $c_r$ and $c_q$.
%\frac{n}{\mom^{1/(s+1)}}
\begin{equation}
r = \frac{c_r}{\eps^{2}\cdot \delta}\cdot {\min\left\{\frac{n}{\mom^{1/(s+1)}}\;,\; 2^s\cdot n\cdot \log^2n \cdot\left(\frac{\alpha}{\mom}\right)^{1/s} \right\}} \;,
\label{eq:cond-r-arb}
\end{equation}
\begin{equation}
q = \frac{c_q}{\eps^2\cdot \delta^3}\cdot \min\left\{\frac{n\cdot \alpha}{\mom^{1/s}}, \frac{n^s\cdot \alpha}{\mom},n^{1-1/s}, \frac{n^{s-1/s}}{\mom^{1-1/s}}\right\} \;.\label{eq:cond-q-arb}
\end{equation}
%
% \mnote{D: do we really need $\eps^{-3}$ rather than $\eps^{-2}$? Maybe current analysis allows
% to remove dependence on $\eps$ in $\theta$?}

%\mnote{D: I think we should consider adding at the end how we can get better behavior of $q$ for small $\alpha$ and give a pointer from here.}
Clearly the setting of $r$ and $q$ in Equation~\eqref{eq:cond-r-arb} and Equation~\eqref{eq:cond-q-arb} respectively, {can only improve} on the setting of $r$ and $q$ for the general case in Equation~\eqref{eq:cond} (Section~\ref{sec:GRS}).
% By Claim~\ref{clm:alpha-mom}, the setting of $r$ in Equation~\eqref{eq:arb-cond}  can only improve on the setting of $r$ in Equation~\eqref{eq:cond}.
% (up to poly-logarithmic factors in $n$ and the factor of  $2^s$).

Our main challenge is in proving that Condition~\ref{cond:r} holds for
$r$ as set in Equation~\eqref{eq:cond-r-arb} (when the graph
has degeneracy at most $\alpha$).
% To this end we
% (re-)define  $ \theta \eqdef {2\alpha^{1/s^2} \mu^{(s-1)/s^2}_s}$ %/{\eps^{1/s}}$
% and let
% $\mH \eqdef \{u: d_u > \theta\}$ and $\mL\eqdef V\setminus \mH$, as in \Sec{GRS}.
% {Similarly to what was shown in \Sec{GRS}, this partition of the vertices into $\mH$ and $\mL$ (based on the degree threshold $\theta$) will aid us in upper bounding $\sum_v \wt(v)^2$.
Here too, the goal is to upper bound $\sum_v \wt(v)^2$.
However, as opposed to the proof of Lemma~\ref{lem:cond-r} in Section~\ref{sec:GRS},
	where we simply obtained an upper bound on $\max_v\{\wt(v)\}$ (and bounded
	$\sum_v \wt(v)^2$ by $\max_v\{\wt(v)\}\cdot \mom$), here the analysis is more refined,
	and uses the degeneracy bound.
%	\ifnum\soda=0
	For details see the proof of our main lemma, stated next.
\iftrue

\begin{lemma}[Condition~\ref{cond:r} holds] \label{lem:cond-r-arb}
	For a sufficiently large constant $c$,
	$$\sum_v \wt(v)^2 \leq c\cdot %\eps^{-1}\cdot
	2^s\cdot  \alpha^{1/s}\cdot  \mom^{2-1/s}\cdot \log^2n\;.$$
\end{lemma}
In order to prove the lemma we first introduce the following definitions and claim.

\begin{definition}\label{def:Gamma+}
	For a set $S$ and a vertex $u$, let $\Gamma_S(u)$ denote the set $\Gamma(u) \cap S$, and let $\Gamma^+_S(u)$ denote the set $\Gamma^+(u) \cap S$. For two sets of vertices $S$ and $T$ (which are not necessarily disjoint),
	let $E^+(S,T) \eqdef \{(u,v): (u,v)\in E^+, u\in S, v\in T\}$.
\end{definition}

\begin{definition}\label{def:U_i}
	We partition the vertices (with degree at least $1$) according to their degree.
	For $0 \leq i \leq \lceil \log n \rceil$, let
	$$U_i \eqdef \{u \in V: d_u \in (2^{i-1},2^i]\} \;.$$
	For each $i$ we partition the vertices in $V$ according to the number of outgoing edges that
	they have to $U_i$. Specifically,
	for $1 \leq j \leq \lceil \log(n/ \alpha) \rceil$, define
	$$V_{i,j} \eqdef \Big\{v \in V : \  |\Gamma^+_{U_i}(v)| \in \left(2^{j-1}\alpha,2^j\alpha\right]\Big\}\;.$$
	Also define $V_{i,0} \eqdef \Big\{v \in V : \ |\Gamma^+_{U_i}(v)| \leq \alpha\Big\}$.
	Hence, $\{V_{i,j}\}_{j=0}^{\lceil \log(n/\alpha)\rceil}$ is a partition of $V$ for each $i$.
\end{definition}
A central building block in the proof of Lemma~\ref{lem:cond-r-arb} is the next claim.
This claim establishes an upper bound on the number of edges going from vertices in $V_{i,j}$ to vertices
in $\whU$, for every $i$ and $j$ (within the appropriate intervals). 
In the proof of this claim we exploit the degeneracy bound.

\begin{claim} \label{clm:E+_U_V}
	Let $V_{i,j}$ and $\whU$ be as defined in Definition~\ref{def:U_i}, and let $E^+(V_{i,j}, \whU)$ be as defined in Definition~\ref{def:Gamma+}. For every $ 0 \leq i \leq \lceil \log n\rceil$ and every $2 \leq j \leq \lceil \log (n/\alpha) \rceil$,
\ifnum\conf=0
\[|E^+(V_{i,j},\whU)| \leq \mom\cdot 2^{-((i-1)(s-1)+j-1)}\;.\]
\fi
\ifnum\conf=1
$|E^+(V_{i,j},\whU)| \leq \mom\cdot 2^{-((i-1)(s-1)+j-1)}\;.$
\fi
\end{claim}
\begin{proof}
	Since $G$ has degeneracy at most $\alpha$,
	\old{by Nash-Williams' equality (Theorem~\ref{thm:Nash-Williams}),}
\ifnum\conf=0
	\[ |E^+(V_{i,j},U_i)| = |E^+(V_{i,j},\whU)|
	= |E(V_{i,j},\whU)| \leq \alpha\cdot(|V_{i,j}| + |\whU|) \;.\]
\fi
\ifnum\conf=1
	$ |E^+(V_{i,j},U_i)| = |E^+(V_{i,j},\whU)|
	= |E(V_{i,j},\whU)| \leq \alpha\cdot(|V_{i,j}| + |\whU|) \;.$
\fi
On the other hand, by the definition of $V_{i,j}$,
\ifnum\conf=0
	\[|E^+(V_{i,j},\whU)| \geq 2^{j-1}\cdot \alpha\cdot| V_{i,j}|\;.\]
\fi
\ifnum\conf=1
$|E^+(V_{i,j},\whU)| \geq 2^{j-1}\cdot \alpha\cdot| V_{i,j}|\;.$
\fi
	Combining the above two bounds (and because  $2^{j-1} - 1 \geq 1$ for $j \geq 2$), we get that
	$|\whU| \geq |V_{i,j}|$,
	and we obtain the following bound on the number of edges in $E^+$ between $V_{i,j}$ and
	their neighbors in $U_i$:
	\begin{equation}
	|E^+(V_{i,j},\whU)| \leq 2\cdot \alpha\cdot |\whU|\;.
	\label{eq:ELij-Hi-ub}
	\end{equation}
	We next upper bound $|\whU|$.
	By the definition \ttcom{of $V_{i,j}$, for every $v\in V_{i,j}$ there exists a vertex $u \in U_i$ such that $d_u \geq d_v$. By the definition of $U_i$, for every $u', u'' \in U_i$, $d_{u'}\geq d_{u''}/2$, implying that for every $v\in V_{i,j}$ and every $u' \in U_i$ it holds that $d_{u'}\geq d_v/2$. Hence, for every $u \in U_i$, we have that $d_u \geq 2^{j-2}\alpha$.
	Also by the definition of $U_i$, for every $u\in U_i$ it holds that $d_u \geq 2^{i-1}$.}
	Therefore,
\ifnum\conf=0
	$$
	\mom \geq \sum_{u \in \whU} d^s_u \geq |\whU| \cdot 2^{(i-1)(s-1)} \cdot 2^{j-\ttcom{2}}\alpha\;,
	$$
\fi
\ifnum\conf=1
$ \mom \geq \sum_{u \in \whU} d^s_u \geq |\whU| \cdot 2^{(i-1)(s-1)} \cdot 2^{j-\ttcom{2}}\alpha\;,$
\fi
	which directly implies that
	\begin{equation}
	|\whU| \leq \mom\cdot 2^{-((i-1)(s-1)+j-2)}\cdot \alpha^{-1}\;.
	\label{eq:mHij-ub}
	\end{equation}
	The proof follows by plugging in Equation~\eqref{eq:mHij-ub} in Equation~\eqref{eq:ELij-Hi-ub}.
\end{proof}

We are now ready to prove Lemma~\ref{lem:cond-r-arb}.
% \begin{proof}[

%\medskip\noindent{\sf Proof of Lemma~\ref{lem:cond-r-arb}.}~
\ifnum\lipics=1
\begin{proof}[Proof of Lemma ~\ref{lem:cond-r-arb}]
\else
\begin{proofof}{Lemma ~\ref{lem:cond-r-arb}}
\fi
	By the definition of $\wt(v)$, and since $d_v \leq d_u$ for every $v$ and $u \in \Gamma^+(v)$,
	\begin{equation}
	\sum_v \wt(v)^2  = \sum_{v} \Bigg(\sum_{u \in \Gamma^+(v)} \left(d_v^{s-1} + d_u^{s-1}\right) \Bigg)^2
	\;\leq\; 4\cdot  \sum_v   \Bigg(\sum_{u \in \Gamma^+(v)}   d_u^{s-1} \Bigg)^2 \;.
	\label{eq:var-one}
	\end{equation}
	In order to bound the expression on the right-hand-side of Equation~\eqref{eq:var-one}
	we consider each $U_i$ (as defined in Definition~\ref{def:U_i}) separately:
	By applying H\"{o}lder's inequality	with conjugates $s$ and $s/(s-1)$
	we get the following bound for every $v$ \ttcom{and $0 \leq i \leq \lceil \log n\rceil$}.
	\begin{equation}\sum_{u \in \Gamma^+_i(v)} 1\cdot d^{s-1}_u
	\leq |\Gamma^+_i(v)|^{1/s}\cdot  \Big(\sum_{u \in \Gamma^+_i(v)}d^s_u\Big)^{(s-1)/s}
	\leq |\Gamma^+_i(v)|^{1/s} \cdot \mom^{(s-1)/s}\;.
	\label{eq:sum-Hi}
	\end{equation}
	For a vertex  $u$,  let $\Gamma^-(u) \eqdef \{v: u\in \Gamma^+(v)\}$.
	By applying Equation~\eqref{eq:sum-Hi} (to one term of the square $\left(\sum_{u\in \Gamma^+_i(v)} d_u^{s-1}\right)^2$),
	\begin{eqnarray}
	\sum_v \left(\sum_{u\in \Gamma^+_i(v)} d_u^{s-1}\right)^2
	& \leq & \sum_v  |\Gamma^+_i(v)|^{1/s} \cdot \mom^{(s-1)/s}
	\cdot \sum_{u\in \Gamma^+_i(v)} d_u^{s-1}  \nonumber \\
	& = & \mom^{(s-1)/s} \cdot \sum_{u \in U_i} d^{s-1}_u \cdot \sum_{v \in \Gamma^-(u)} |\Gamma^+_i(v)|^{1/s} \label{eq:switch} \\
	& = & \mom^{(s-1)/s} \cdot \sum_{j=0}^{\lceil \log n \rceil}
	\Big(\sum_{u \in U_i} d^{s-1}_u \cdot \sum_{v \in \Gamma^-(u)\cap V_{i,j} } |\Gamma^+_i(v)|^{1/s}\Big)\;. \label{eq:split}
	\end{eqnarray}
	Equation~\eqref{eq:switch} follows by the definition of $\Gamma^-(u)$ (and switching the order of the summations), and Equation~\eqref{eq:split} follows from splitting the sum in Equation~\eqref{eq:switch} based on the partition of $V$ into
	the subsets $V_{i,j}$ (recall that $i$ is fixed for now \ttcom{and $V_{i,j}$ is as defined in Definition~\ref{def:U_i}}).
	
	From this point on we only consider $j$'s such that $V_{i,j}$ is not empty. For each $j$ (and $i$), by the definition of $V_{i,j}$,
	\begin{eqnarray}
	\sum_{u \in U_i} d^{s-1}_u \sum_{v \in \Gamma^-(u) \cap V_{i,j}} |\Gamma^+_i(v)|^{1/s}
	& \leq & \sum_{u \in U_i} d^{s-1}_u \cdot \sum_{v \in \Gamma^-(u)\cap V_{i,j} } (2^{j} \cdot \alpha)^{1/s}\nonumber \\
	& \leq & 2^{j/s} \cdot \alpha^{1/s} \cdot \sum_{u \in U_i} d^{s-1}_u \cdot |\Gamma^-(u)\cap V_{i,j}|\;.
	\label{eq:d-s-1-u-Gamma-i}
	\end{eqnarray}
	For $j < 2$, we trivially upper bound $|\Gamma^-(u)\cap V_{i,j}|$ by $d_u$.
	Thus,
	% the expression in Equation~\eqref{eq:d-s-1-u-Gamma-i}
	% can be upper bounded by
	\begin{equation}
	\sum_{u \in U_i} d^{s-1}_u \sum_{v \in \Gamma^-(u) \cap V_{i,j}} |\Gamma^+_i(v)|^{1/s}
	\;\leq\; 2^{1/s} \cdot \alpha^{1/s} \cdot\sum_{u \in U_i} d^s_u \leq 2\cdot \alpha^{1/s} \cdot\mom
	\;\;\;\;\;\;\;\; \mbox{(for $j<2$)}\;.
	\label{eq:j<2}
	\end{equation}
	Turning to $j \geq 2$,
	since all vertices in $U_i$ have degree at most $2^i$ and by Equation~\eqref{eq:d-s-1-u-Gamma-i}, we get:
	\begin{eqnarray}
	\sum_{u \in U_i} d^{s-1}_u \cdot \sum_{v \in \Gamma^-(u)\cap V_{i,j}} |\Gamma^+_i(v)|^{1/s}
	&\leq& 2^{j/s} \cdot \alpha^{1/s} \cdot 2^{i(s-1)} \cdot |E^+(V_{i,j},U_i)|. \nonumber
	\end{eqnarray}
	Hence, by Claim~\ref{clm:E+_U_V},
	\begin{eqnarray}
	\sum_{u \in U_i} d^{s-1}_u \cdot \sum_{v \in \Gamma^-(u)\cap V_{i,j}} |\Gamma^+_i(v)|^{1/s} &\leq& 2^{j/s -j +s}\cdot \alpha^{1/s}\cdot \mom\
	\;.	\label{eq:j>2-main}
	\end{eqnarray}
	By using the bound in Equation~\eqref{eq:j<2} for $j<2$, the bound in
	Equation~\eqref{eq:j>2-main} for $j\geq 2$
	and plugging them in Equation~\eqref{eq:split} we get
	\begin{eqnarray}
	\sum_v \left(\sum_{u\in \Gamma^+_i(v)} d_u^{s-1}\right)^2
	&\leq& 4\cdot \alpha^{1/s} \cdot\mom^{2-\old{2}1/s} +
	2^{s} \cdot \alpha^{1/s}\cdot \mom^{2-1/s}\cdot \sum_{j=3}^{\lceil\log n\rceil} 2^{j/s-j} \nonumber\\
	&\leq& 2^{s+1}\cdot \alpha^{1/s} \cdot\mom^{2-\old{2}1/s} \cdot \log n\;,
	\label{eq:all-j}
	\end{eqnarray}
	where the last inequality holds because $s\geq 1$
	(in fact, for $s> 1$ we can save a $\log n$ factor).
	Using the inequality $(\sum_{i=1}^{\ell} x_i)^2 \leq (2\ell-1)\cdot \sum_{i=1}^\ell x_i^2$, Equation~\eqref{eq:all-j} implies that
	\begin{equation}
	\sum_v \left(\sum_{u\in \Gamma^+(v)} d_u^{s-1}\right)^2
	\leq 2^{s+2} \cdot \alpha^{1/s}\cdot \mom^{2-1/s}\cdot  \log^2 n\;.
	\label{eq:second-term}
	\end{equation}
	The lemma follows by combining Equation~\eqref{eq:var-one} with %Equations~\eqref{eq:first-term}
	Equation~\eqref{eq:second-term}.
\ifnum\lipics=1
\end{proof}
\else
\end{proofof}
\fi

\fi %end proof of lemma cond-1-arb

\ifnum\soda=0
It remains to establish Condition~\ref{cond:q}.
\else
The next lemma, which establishes Condition~\ref{cond:q}, can be proved similarly to Lemma~\ref{lem:cond-q}.
\fi
\begin{lemma}[Condition~\ref{cond:q} holds] \label{lem:cond-q-arb}
	\[\min\left\{\frac{n \cdot \alpha }{\mom^{1/s}}, \frac{n^s \cdot \alpha}{\mom}, n^{1-1/s},\frac{n^{s-1/s}}{\mom^{1-1/s}}  \right\} \geq m \cdot \frac{\momts}{\mom^2} \;.\]
\end{lemma}
\ifnum\soda=0
\begin{proof}
	Since $G$ has bounded degeneracy $\alpha$, it follows that $m \leq n \cdot\alpha$.
	{By Equation~\eqref{eq:m-n-mu}}, $m\leq n^{1-1/s}\cdot \mom^{1/s}$.
	% Also recall that by Lemma~\ref{lem:cond3},
	{As shown in the proof of Lemma~\ref{lem:cond-q},}
	$\momts \leq \min\{\mom^{2-1/s},n^{s-1}\cdot \mom\}$. The proof follows from the above two bounds.
\end{proof}
\fi

%%%%%%%%%%%%%%%%
\ifnum\conf=0
\subsection{The case of $\mathbf{s=1}$: estimating the average degree}\label{subsec:s=1}
When $s=1$ (so that $\mom = \momone = 2m$ and $\avgmom_s=\avgmom_1$ is the average degree), there is a very simple analysis of a slight variant of
\mainalg$_s$. Observe that for $s=1$, by Definition~\ref{def:wt_e},  for every edge $e$, $\wt(e) = 2$, and $\wt(v) = 2\cdot d^+_v$.
% Also observe that $\theta = 2\alpha/\eps$ for $s=1$.
% Recalling that $\mH \eqdef \{v: d_v > \theta\}$ (and $\mL \eqdef V\setminus \mH$),
For a degree threshold $\theta = 2\alpha/\eps$, let $\mH \eqdef \{v: d_v > \theta\}$, and $\mL \eqdef V\setminus \mH$.
By the definition of $\mH$ we have that
$|\mH| < \momone/\theta = \eps\momone/(2\alpha)$. Since the graph has degeneracy at most $\alpha$,
$\sum_{v\in \mH} d^+_v \leq \alpha \cdot |\mH| \leq \eps\momone/2$.
This implies that $\sum_{v\in \mL} \wt(v) \geq (1-\eps)\momone$.

Suppose that we modify the algorithm so that $X_i$ is set to
$d^{s-1}_{v_i} + d^{s-1}_{u_i} = 2$ only if $d_{v_i} \leq  \theta$ (as well as $v_i \prec u_i$),
and is otherwise set to $0$. Under this modification, $\EX[X] \in [(1-\eps)\momone,\momone]$.
Since $\wt(v) \leq 2\theta$ for each $v\in \mL$, we get that
$\sum_{v\in \mL} \wt(v)^2 \leq 2\theta\cdot \momone = (4\alpha/\eps)\cdot \momone$.
Therefore, in order to satisfy Condition~\ref{cond:r}, it suffices to set
$r = \frac{c_r \cdot n \cdot \alpha}{\eps^3\cdot \delta \momone}$.
Thus, as compared to the setting in Equation~\eqref{eq:cond-r-arb},
we save a $\log^2 n$ factor (at the cost of factor of $1/\eps$), but, more importantly,
the analysis is very simple (as compared to the proof of Lemma~\ref{lem:cond-r-arb}).
%  $r \geq (120 \cdot n \cdot \sum_v \wt(v)^2)/(\eps^2 \cdot \delta \cdot \mom^2)$
The setting of $q$ is as in Equation~\eqref{eq:cond-q-arb}, which for $s=1$ gives
$q= \frac{c_q}{\eps^2\cdot \delta^3}$.
% In particular, if $m = \Omega(n)$ and $\alpha$ is a constant, the complexity
% of the algorithm is $\poly(1/\eps,1/\delta)$.
\fi
%%%%%%%%%%%%%%

% \input{wrapping_things_sec}
\section{Wrapping things up} \label{sec:search_no_alpha}

%%%%%
The proof of our final result, Theorem~\ref{thm:correctness}, follows by combining
Theorem~\ref{thm:cond}, Lemma~\ref{lem:cond-r}, %Lemma~\ref{lem:cond-q},
% D: we need Lemma~\ref{lem:cond-r} because of the minimum on the settings of $r$
Lemma~\ref{lem:cond-r-arb} and Lemma~\ref{lem:cond-q-arb}, with a geometric search for a  factor-$2$
estimate of
$\mom$ (which determines the correct setting of $r$ and $q$ in the algorithm).
%%%%

\ifnum\conf=0
\medskip
For convenience, we restate the bounds of Theorem~\ref{thm:correctness} in terms on $\mom$.
\begin{eqnarray*}
	O\Big(2^s\cdot n\cdot\log^2 n\cdot \Big(\frac{\alpha}{\mom}\Big)^{1/s} +\min\Big\{\frac{n\cdot \alpha}{\mom^{1/s}}, \frac{n^{s}\cdot  \alpha}{\mom}\Big\}\Big)\cdot \;\frac{s \log n\cdot \log(s\log n)}{\eps^2} %\label{eq:arb-rt-full}
\end{eqnarray*}
\begin{eqnarray*}
	O\Big(\frac{n}{\mom^{1/(s+1)}} +
	\min\Big\{n^{1-1/s}, \frac{n^{s-1/s}}{\mom^{1-1/s}}\Big\}\Big)\cdot\frac{s\log n\cdot \log(s\log n)}{\eps^2}
	%  \label{eq:grs-rt-full}
\end{eqnarray*}

\ifnum\lipics=1
\begin{proof}[Proof of Theorem~\ref{thm:correctness}]
\else
\begin{proofof}{Proof of Theorem~\ref{thm:correctness}}
\fi
	Recall that the setting of $r$ and $q$ in Equation~\eqref{eq:cond-r-arb} and Equation~\eqref{eq:cond-q-arb}, respectively,
	equals the setting in Equation~\eqref{eq:cond} when $\alpha$ is set to its maximum possible value
	$\mom^{1/(s+1)}$. Hence, it suffices to prove the theorem under the assumption that the
	algorithm is provided with $\alpha$ (which upper bounds $\alpha(G)$).
	It follows from Theorem~\ref{thm:cond}, Lemma~\ref{lem:cond-r}, %Lemma~\ref{lem:cond-q},
	% D: we need Lemma~\ref{lem:cond-r} because of the minimum on the settings of $r$
	Lemma~\ref{lem:cond-r-arb} and Lemma~\ref{lem:cond-q-arb}, that when \mainalg$_s$ is invoked with parameters $r$ and $q$ as set
	% in Equation~\eqref{eq:cond} for general graphs,
	in Equation~\eqref{eq:cond-r-arb} and Equation~\eqref{eq:cond-q-arb}, respectively, %for graphs with bounded arboricity $\alpha$,
	the algorithm returns a value $X$ such that
	$X \in [(1-\eps)\avgmom_s,(1+\eps)\avgmom_s]$ with probability at least $1-\delta$.
	However, these settings require the knowledge of $\mom$, which is the parameter we are trying to approximate (up to the normalization factor $n$). Hence, we use the following search algorithm.
	
	We start with a guess $\gm = n^{s+1}$ (the maximum possible value of $\mom$), and compute $r$ according to %Equation~\eqref{eq:cond} in the general case and
	Equation~\eqref{eq:cond-r-arb}  and  $q$ according to Equation~\eqref{eq:cond-q-arb} %in the bounded arboricity case,
	assuming $\mom = \gm$, with the given approximation parameter $\eps$ and with
	%$\delta = \frac{1}{10\log n}$.
	$\delta = 1/3$.
	We then invoke \mainalg$_s$  $\Theta(\log(s\log n))$ times
	% For each $\Theta({\rm loglog} n)$ consecutive invocations we take the median value,
	and let $Z$ be the median of the returned values. %value over the computed medians.
	If $Z \geq \gm$ then we stop and return $Z$. Otherwise we halve $\gm$ and repeat.
	
	Observe that $r$ and $q$ are decreasing functions of $\mom$.
	Hence, the running time of each invocation of \mainalg$_s$ is at most the running
	time of the last invocation. By  Claim~\ref{clm:exp} and
	Markov's inequality, for each invocation of \mainalg$_s$,
	$\Pr[X \geq 3\mom]\leq 1/3$ (where $X$ is the value returned by the algorithm).
	We stress that this holds regardless of whether $r$ and $q$ satisfy Conditions~\ref{cond:r} and~\ref{cond:q} (respectively).
	By the definition of $Z$, for values $\gm > 3\mom$, the probability that we will stop in each step
	is $O(1/(s\log n))$. Therefore, with high constant probability we will not stop before $\gm \leq 3\mom$.
	
	Once $\gm  \leq 3\mom$, we will satisfy\footnote{To be precise, in order to satisfy the conditions for values $\mom \leq 3\gm$, we invoke \mainalg$_s$ with $3r$ and $3q$.}
	Conditions~\ref{cond:r} and \ref{cond:q}.
	By Theorem~\ref{thm:cond}, in each invocation of \mainalg$_s$, $X \in [(1-\eps)\avgmom_s, (1+\eps)\avgmom_s]$ with probability at least $1-\delta = 2/3$, so that $Z  \in [(1-\eps)\avgmom_s, (1+\eps)\avgmom_s]$
	with probability at least $1- O(1/(s\log n))$.
	% Hence the median value over $\Theta({\rm \log\log}n)$
	Thus, once $\gm \leq \mom/2$, the algorithm will stop and return a value in
	$ [(1-\eps)\avgmom_s, (1+\eps)\avgmom_s]$
	with probability at least $1- O(1/(s\log n))$. By a union bound over all iterations, the algorithm
	returns such a value with high constant probability.
	
	In order to bound  the expected running time, we first observe that the running time of an invocation of \mainalg$_s$ with $\gm \in [\mom/2^{i+1},\mom/2^i)$ is at most $2^{i+1}$ times the running
	time of an invocation with $\gm =\mom$.
	On the other hand, the probability that the algorithm does not stop before
	we reach $\gm \in [\mom/2^{i+1},\mom/2^{i})$ for any  $i\geq 1$, is $O(\log(n)^{-i})$.
	The bound on the expected running time follows.
	% $1-\Theta(\log \log n \cdot \delta)$.
	%By the setting of $\delta$, with probability at least $2/3$, the algorithm returns a $(1\pm\eps)$-estimate of $\mom$.
	% The running time is dominated by the running time of the last invocation, which is $O(r+q)$.
\ifnum\lipics=1
\end{proof}
\else
\end{proofof}
\fi

\fi
%%%%%%%%%%%%%%%

%We note that if we set $\alpha=\mom^{1/(s+1)}$ and invoke \mainalg$_s$ with parameters $r$ and $q$ as set %in Equation~\eqref{eq:arb-cond}, we achieve the running time stated in Theorem~\ref{thm:correctness} for general graphs (up to the %dependence on $2^s$ and $\log n$).

%%%%%%%%
\ifnum\conf=0
\section{Lower Bounds for Bounded Degeneracy}

The lower bounds given in this section hold for algorithms that are allowed degree and neighbor queries,
as well as pair queries (that is, queries of the form: ``is there an edge between $u$ and $v$'').
{These lower bound show that the complexity of the algorithm, as stated in Theorem~\ref{thm:correctness}
	for graphs with degeneracy at most $\alpha$,
	is tight up to the dependence on $1/\eps$ and polylogarithmic factors in $n$, for any constant $s$
	(or even $s = O(\log\log n)$).}
% \mnote{D: say something about $s$ being a constant?}

\begin{theorem} \label{thm:lb_first_term}
	Any  constant-factor approximation algorithm for $\mom$
	% $\sum_{v \in V}d_v^s$
	% of graphs with $n$ vertices and arboricity $\alpha$
	must perform $\Omega\left(\frac{n\cdot \alpha(G)^{1/s}}{\mom^{1/s}(G)}\right)$ queries.
	% when $\alpha > \left(\frac{\mom(G)}{n}\right)^{1/s}$,\mnote{D: is this really what we need?}	
\end{theorem}

\begin{proof}
	For every $n$, $\wmom$ and
	$\wmom/n^s \leq \walpha \leq (\wmom/c)^{1/(s+1)}$,\footnote{Recall that by Claim~\ref{clm:alpha-mom},
		$\alpha(G) \leq \mom(G)^{\frac{1}{s+1}}$, and note that $\mom(G) = \sum_v d_v^s \leq \sum_v d_v\cdot n^{s-1}
		\leq 2\alpha(G) n^s$.}
	we next define two families of graphs: $\mG_1$ and $\mG_2$.  % such that in both families all graphs have arboricity $\alpha$, but such that the.
	% $G_1$ and a family of graphs $\mG_2$ for which the following holds.
	Each graph in $\mG_1$ consists of a clique $C_1$, over $\walpha$ vertices, and an
	independent set $C_2$, over $n-\walpha$ vertices. For each graph in $\mG_2$, the vertices are partitioned into three sets:
	$C_1, C_2, C_3$, of sizes $\walpha$, $(\wmom/\walpha)^{1/s} - \walpha$, and $n-(\wmom/\walpha)^{1/s} $,
	respectively. The set $C_1$ is a clique, and is connected by a complete bipartite graph to
	$C_2$, where there are no edges within $C_2$. The set $C_3$ is an independent set.
	% complete bipartite graph $K_{\alpha, (\wmom/\walpha)^{1/s}}$ and
	%	an independent set of size $n-\walpha-(\wmom/\walpha)^{1/s}$.
	Within
	each family, the graphs differ only in the labeling of the vertices. By construction, in both families,
	all graphs have degeneracy $\Theta(\walpha)$.
	For each $G\in \mG_1$ we have that
	$\mom(G) = O(\walpha^{s+1}) \leq \wmom/c$,
	and for each $G$ in $\mG_2$,
	$\mom(G) = \Theta(\walpha \cdot (\wmom/\walpha) + ((\wmom/\walpha)^{1/s}-\walpha) \cdot \walpha^s) = \Theta(\wmom)$.
	
	Clearly, unless the algorithm ``hits'' a vertex in the clique $C_1$ of a graph belonging to $\mG_1$ or a vertex in $C_1\cup C_2$ in a graph belonging to $\mG_2$, it cannot distinguish between
	a graph selected randomly from $\mG_1$ and a graph selected randomly from $\mG_2$.
	The probability of hitting such a vertex % in a graph selected uniformly at random from 	$\mG_2$
	is $O\left( \frac{(\wmom/\walpha)^{1/s}}{n}\right)$. Thus, in order for this event to occur with high constant probability, $\Omega\left(\frac{n\walpha^{1/s}}{\wmom^{1/s}}\right)$ queries are necessary.
\end{proof}

\begin{theorem}\label{thm:lb_second_term}
	Any constant-factor approximation algorithm for $\mom$ %$\sum_{v \in V}d_v^s$
	must perform
	$$\Omega\left(\min\left \{\frac{n \cdot \alpha(G)}{\mom(G)^{1/s}}, \frac{n^s \cdot \alpha(G)}{\mom(G)},  n^{1-1/s}, \frac{n^{s-1/s}}{\mom(G)^{1-1/s}} \right \}  \right)$$ queries.
	% where the allowed queries are degree queries, neighbor queries and pair queries.
\end{theorem}

%
%{\setlength{\tabulinesep}{5pt}
%\begin{tabu}{ | c |[2pt] c | c | }
%	\hline				
%	& $\alpha < (\mom/n)^{1/s}$      & $\alpha \geq (\mom/n)^{1/s}$  \\	\tabucline[2pt]{-}
%	 $\mom^{1/s} \leq n$ & $n\alpha/\mom^{1/s}$ &  $ n\alpha^{1/s}/\mom^{1/s} + n^{1-1/s}$	 \\ \hline
%	 $\mom^{1/s} > n$ & $ n^2\alpha/\mom$ & $n\alpha^{1/s}/\mom^{1/s} +  n^{s-1/s}/\mom^{1-1/s}$ \\
%	\hline
%\end{tabu}
%}

%%%%%%%%%
\iffalse
\begin{theorem}\label{thm:lb_second_term}
	The following bounds hold for constant-factor approximation algorithms for $\mom$.
	% that may perform degree and neighbor queries. \mnote{D: maybe also pair queries?}
	\begin{enumerate}
		\item If $\alpha(G) \geq (\mom(G)/n)^{1/s}$ and $\mom(G)^{1/s}\leq n$,
		then any constant-factor approximation algorithm for $\mom(G)$ must perform $\Omega\left(n^{1-1/s}\right)$ queries. \label{item:lb_item_1}
		\item If $\alpha(G) \geq (\mom(G)/n)^{1/s}$ and $\mom(G)^{1/s} > n$, then then any constant-factor approximation algorithm for $\mom(G)$ must perform $\Omega\left(n^{s-1/s}/\mom(G)^{1-1/s} \right)$ queries.\label{item:lb_item_2}
		\item If $\alpha(G) < (\mom(G)/n)^{1/s}$ and $\mom(G)^{1/s}\leq n$,
		then any constant-factor approximation algorithm for $\mom(G)$ must perform $\Omega\left(n \alpha/\mom(G)^{1/s} \right)$ queries. \label{item:lb_item_3}
		\item 	If $\alpha(G) < (\mom(G)/n)^{1/s}$ and $\mom(G)^{1/s} > n$, then and constant-factor approximation algorithm for $\mom(G)$  must perform $\Omega\left(n^s \alpha/\mom(G) \right)$ queries.\label{item:lb_item_4}
	\end{enumerate}
\end{theorem}
\fi
%%%%%%%%%%%%%%

\begin{proof}
	The proof of the theorem is based on simple modifications to the lower bound constructions in~\cite[Thm.~5]{GRS11}. We note that this theorem of~\cite{GRS11} was stated explicitly for algorithms that perform degree and neighbor queries since such was the algorithm presented in~\cite{GRS11} (as well as the current paper). However, it was noted in~\cite[Sec. 7]{GRS11} that they also hold when pair queries are allowed (as they are essentially based on ``hitting special vertices'').
	
	The theorem is proved by considering two cases that are defined according to the relation between $\mom(G)$ and $n$
	(namely, $\mom(G)^{1/s} \leq n - c$ and $\mom(G)^{1/s}> n -c$ for a constant $c$, which determines the hardness of approximation). Within each case there are two sub-cases that are defined according to the relation between $\alpha(G)$ and $(\mom(G)/n)^{1/s}$
	(namely, $\alpha(G) < (\mom(G)/n)^{1/s}$ and $\alpha(G) \geq (\mom(G)/n)^{1/s}$). Each of the four sub-cases gives us one of the terms in the lower bound (within the $\min\{\cdot\}$ expression).
	For each of the sub-cases we consider the values $\walpha$ and $\wmom$ that correspond to the sub-case, and we construct two families of graphs: $\mG_1$ and $\mG_2$. In both families all graphs have degeneracy $\Theta(\walpha)$. Every graph $G \in \mG_1$
	satisfies $\mom(G) \leq \wmom$, while every graph $G \in \mG_2$
	satisfies $\mom(G) \geq  c\cdot\wmom$. % for a constant $c$.
	The lower bound is based on the difficulty of
	distinguishing between a random graph selected from $\mG_1$ and a random graph selected from $\mG_2$,
	with a specified number of queries.
	In all cases we modify the construction in~\cite[Thm. 5]{GRS11} either by decreasing the degeneracy or increasing  it.
	
	\medskip\noindent{\sf \boldmath \bfseries The case $\wmom^{1/s}\leq n -c$.~}
	For this case we modify the construction described in~\cite[Item~(2) of Thm.~5]{GRS11}.
	In both families, the vertices are partitioned into two subsets, $S$ and $V\setminus S$, where
	the size of $S$ is  $c$. % (which determines the gap between the number of stars in both families).
	For each graph in $\mG_1$, the set $S$ is an independent set, while for each graph in $\mG_2$, each
	vertex in $S$ has $d'$ neighbors in $V\setminus S$, for an appropriate setting of $d'$. In both families, the vertices in $V\setminus S$ have degree $d$, for an appropriate setting of $d$ (where in $\mG_1$ their neighbors are all in $V\setminus S$ while in $\mG_2$ some of their neighbors are in $S$).
	Observe that the graphs in the family $\mG_1$ can be viewed as obtained from the graphs of $\mG_2$ be replacing pairs of edges between $S$ to $V/S$ by a single edge in $V\setminus S$. We refer to the edges between $S$ and $V\setminus S$ in the graphs of $\mG2$, and the edges replacing them in the graphs of $\mG_1$ as ``special edges''.
	In~\cite{GRS11}, the settings of $d'$ and $d$ are such that the number of stars in graphs belonging to $\mG_2$ is (roughly) a factor $c$ larger
	than the number of stars in graphs belonging to $\mG_1$. The difficulty of distinguishing between a random
	graph selected from $\mG_1$ and a random graph selected from $\mG_2$ is based on upper bounding
	the following two very similar events: (1) ``Hitting'' a vertex $v$ in $S$ when querying a graph in $\mG_2$ by either performing a degree/neighbor query on $v$ or by performing a neighbor query on $u \in V\setminus S$ and receiving $v$ as an answer. (2) ``Hitting'' a vertex $v$ in $S$  or a special edge between vertices in $V\setminus S$ when querying a graph in $\mG_1$, where the number of special edges is $|S|\cdot d'/2$. \tout{(these edges are viewed as corresponding to edges in graphs belonging to $\mG_2$ between $S$ and $V\setminus S$ that
		were removed and replaced by edges internal to $V\setminus S$ in graphs belonging to $\mG_1$.)}
	In what follows we modify the settings of $d'$ and $d$ (as defined in~\cite[Item~(2) of Thm.~5]{GRS11}), and in the case of large $\alpha$ perform an
	additional small modification.
	
	In both the sub-case $\walpha < (\wmom/n)^{1/s}$ and the sub-case $\walpha \geq (\wmom/n)^{1/s}$,
	we set $d'  =\lceil \wmom^{1/s} \rceil$. This ensures that for each graph $G \in \mG_2$,
	$\mom(G) \geq c \wmom$.
	
	If $\walpha < (\wmom/n)^{1/s}$, then we set $d=\walpha$. Hence, $\alpha(G) = \Theta(\walpha)$ for graphs in both families, and $\mom(G) \leq n \cdot ((\wmom/n)^{1/s})^s = \wmom$
	for graphs in $\mG_1$. Since the number
	of edges between $S$ and $V\setminus S$ in graphs belonging to $\mG_2$ (which is of the same order as
	the number of special edges in graphs belong to $\mG_1$) is $O(\wmom^{1/s})$ while the total
	number of edges is $\Omega(n\cdot \walpha)$, we get a lower bound of $\Omega\left(n\cdot \walpha/\wmom^{1/s} \right)$ (this of course requires formalizing, as done in~\cite{GRS11}).
	
	If $\walpha \geq (\wmom/n)^{1/s}$, then we set $d = \lfloor (\wmom/n)^{1/s}\rfloor$,
	so that it still holds that $\mom(G) \leq n \cdot ((\wmom/n)^{1/s})^s = \wmom$
	for graphs in $\mG_1$.
	In both families, within $V\setminus S$
	we add edges so as to form a clique on a subset of size $\walpha$, thus increasing the degeneracy to
	$\Theta(\walpha)$. Since this modification is the same in both families, it does not effect the ability
	to distinguish between  the two families.
	Since the number of edges (not including those in the clique) is
	$(n-c)\cdot d = \Omega(\wmom^{1/s}\cdot n^{1-1/s})$, we get a lower bound of $\Omega(n^{1-1/s})$.
	
	\medskip\noindent{\sf \boldmath \bfseries The case $\wmom^{1/s}> n -c$.~}
	In this case we may assume, without loss of generality, that $\wmom < n^{s+1}/c'$ for a sufficiently large constant $c'$, or else
	the lower bound $\Omega(n^{s-1/s}/\wmom^{1-1/s})$ is trivial, and similarly that $\wmom < n^s\cdot \walpha/c'$,
	or else the lower bound $\Omega(n^s\cdot \walpha/\wmom)$ is trivial.
	We may also assume that $\walpha \leq \wmom^{1/(s+1)}$ since $\alpha(G) \leq \mom(G)^{1/(s+1)}$
	for every graph $G$ (by Claim~\ref{clm:alpha-mom}).
	
	Here we modify the construction described in ~\cite[Item~(3) of Thm.~5]{GRS11}.
	% \mnote{should it be a coma or a period after Thm?} D: period
	The construction is similar to the one described for $\wmom^{1/s} \leq n-c$, except that the size of the set $S$ needs to be increased.
	Specifically, we let $|S| = b$ for $b = \lceil c\wmom/n^s \rceil$, and in the graphs in $\mG_2$
	% there is a complete bipartite graph between $S$ and $V\setminus S$
	each vertex in $S$ is connected to every other vertex in the graph. Therefore,
	for each $G\in \mG_2$, $\mom(G) = \Omega(\wmom)$.
	
	If $\walpha < (\wmom/n)^{1/s}$, then in both families
	each vertex in $V\setminus S$ has degree $d= \walpha$. Since $b< \walpha$ (due to
	$\wmom < n^s\cdot \walpha/c'$ for a sufficiently large constant $c'$), we get that $\alpha(G)=\Theta(\walpha)$
	for all graphs in $\mG_1$ and in $\mG_2$. The difference between the families is that in
	the graphs belonging to $\mG_2$, each vertex in $V\setminus S$ has $b=|S|$ neighbors in $S$ and $d-b$ neighbors in $V\setminus S$, while in the graphs belonging to $\mG_1$, each vertex in $V\setminus S$ has $d$ neighbors in
	$V\setminus S$ (and each vertex in $S$ only neighbors each other vertex in $S$).
	This implies that for each $G\in \mG_1$,  $\mom(G) \leq n\walpha^s <\wmom$ (where we have again
	used $b < \walpha$). Since the number of edges between $S$ and $V\setminus S$ in each graph in
	$\mG_2$ (the number of corresponding special edges within $V\setminus S$ in each graph in $\mG_1$)
	is $O(b\cdot n) = O(\wmom/n^{s-1})$, and the total number of edges is $\Omega(n\cdot \walpha)$,
	we get a lower bound of $\Omega(n^s\cdot \walpha/\wmom)$.
	
	If $\walpha \geq (\wmom/n)^{1/s}$, then we use the same construction as above only with
	$d = (\wmom/n)^{1/s}$, and we ``plant'' a clique of size $\walpha$ in $V\setminus S$.
	The degeneracy is hence increased to $\Theta(\walpha)$, and the value of $\mom(G)$ for $G\in \mG_1$
	is increased to at most $2\wmom$ (due to the clique).
	Since the number of edges (not including those in the clique) is
	$\Omega(n\cdot d) = \Omega(\wmom^{1/s}\cdot n^{1-1/s})$, we get a lower bound of
	$\Omega(n^{s-1/s}/\wmom^{1-1/s})$.
\end{proof}
\fi

\ifnum\conf=0
%\subsection{Why the algorithm needs the arboricity bound}
\subsection{On knowing the degeneracy bound}
\label{subsec:why_arboricity}

One may wonder if the query complexity of Theorem~\ref{thm:correctness} can be obtained \emph{without} knowledge of the degeneracy $\alpha$.
The ideal situation would be one where an algorithm has the complexity of \mainalg, without knowing $\alpha$.
% Current
The
worst-case lower bound of~\cite{GRS11} do not preclude this possibility, since it uses graphs with high degeneracy.
Nonetheless, a slight adaptation of those arguments shows that the bound holds even for bounded degeneracy graphs,
\emph{if} the algorithm must work on all graphs. We will focus solely on estimating average degree, since that suffices to make our point.

\begin{definition} \label{def:alg}
	% Fix constant $c$.
	% An algorithm $\cA$ is called \emph{valid} if: given input $G$,
	For a constant $c$, an algorithm $\cA$ is called $c$-{\sf valid} if: given query access to a graph $G$,
	with probability $2/3$, $\cA$ outputs a $c$-approximation to the average degree of $G$.
\end{definition}

Note that \mainalg, with access to a degeneracy bound, is not valid. This is because it is only required to be accurate for graphs with the given degeneracy bound, not all graphs.

\begin{theorem} \label{thm:valid-lb} Let $n$ be a sufficiently large integer. Consider the class of graphs on $n$ vertices with degeneracy at most $2$. For any constant $c$, any $c$-valid algorithm
	must perform $\Omega(\sqrt{n})$ queries on these graphs.
\end{theorem}

\begin{proof}
	Similarly to previous lower-bound proofs, for each sufficiently large $n$, we define two distributions over labeled graphs.
	% We begin with the hard distribution and prove some claims about it. We think of the input in terms
	% of a labeled graph. Fix sufficiently large $n$.
	Consider a graph consisting
	of two connected components, a cycle on $\lfloor n-4c\sqrt{n}\rfloor$ vertices and a cycle on $\lceil 4c\sqrt{n} \rceil$ vertices (where $c$ is the constant in the statement of the theorem).
	The distribution $\cG_1$ is generated by labeling the vertices using a uniform random permutation in $[n]$.
	For the second distribution, take the graph that consists of a cycle on $\lfloor n-4c\sqrt{n}\rfloor$ vertices and a clique
	on $\lceil 4c\sqrt{n} \rceil$ vertices. The distribution $\cG_2$ is generated by labeling the vertices according to a uniform random permutation.
	
	Consider any (possibly randomized) algorithm for deciding, given query access either to a graph generated by $\cG_1$ or to a graph generated by $\cG_2$, according to
	which distribution was the graph generated. As long as the algorithm does not perform a query on a vertex that belongs to the small
	cycle (if the graph is generated by $\cG_1$) or the small clique (if the graph is generated by
	$\cG_2$), answers to its queries are identically distributed under the two distribution. This implies that any such decision algorithm
	must perform $\Omega(\sqrt{n})$ queries in order to succeed with probability at least $2/3$.
	
	% \begin{claim} \label{clm:dist} Any procedure that distinguishes inputs from $\cG_1$ from those in
	% $\cG_2$ with probability at least $2/3$ requires $\Omega(\sqrt{n})$ queries.
	% \end{claim}
	
	% \begin{proof} (of Theorem~\ref{thm:valid-lb})
	Now consider a $c$-valid algorithm $\cA$ that makes at most $s$ queries on any graph with $n$ vertices and degeneracy at most $2$.
	We use this algorithm in order to construct an algorithm $\cB$ that distinguishes $\cG_1$ from $\cG_2$ in $O(s)$ queries. It works as follows.
	Given query access to $G$, $\cB$ runs $20$ independent runs of $\cA$. If any run makes more than $s$ queries, $\cB$ terminates
	and outputs ``$\cG_2$". At the end of the runs, all of them have provided some estimate for the average degree. If the median is at most $2c$, then
	$\cB$ outputs ``$\cG_1"$. Otherwise, it outputs ``$\cG_2$".
	
	Suppose $G \sim \cG_1$, and recall that the average degree of $G$ is exactly $2$.
	All runs are guaranteed to make at most $s$ queries. Thus, all of them will output an estimate.
	By the validity of $\cA$ and a Chernoff bound, the median estimate will be at most $2c$ with probability at least $2/3$,
	and $\cB$ gives a correct output. Now suppose that $G \sim \cG_2$, and recall that
	The average degree of $G$ is at least $3c^2$.
	If any run takes more than $s$ queries, then $\cB$ outputs correctly.
	Otherwise, by a similar argument to the one made for $G \sim \cG_1$,
	the median estimate will be at least $3c$ with probability at least $2/3$. Thus, $\cB$ is correct.
	
	The query complexity of $\cB$ is $O(s)$ and it distinguishes $\cG_1$ from $\cG_2$ with probability at least $2/3$. Therefore, $s$ must be $\Omega(\sqrt{n})$.
\end{proof}

\fi
	\bibliography{moments_bib}
		
\end{document}